\documentclass[11p,reqno]{amsart}

\usepackage[T1]{fontenc}
\usepackage{graphicx}
\usepackage{amssymb,amsthm,amsmath,mathrsfs,bm,braket,marginnote}
\usepackage{enumerate}
\usepackage{appendix}
\usepackage[colorlinks=true, pdfstartview=FitV, linkcolor=blue, citecolor=blue, urlcolor=blue]{hyperref}
\usepackage{multirow}

\usepackage{pgf}
\usepackage{pgfplots}
\usepackage{tikz}
\usetikzlibrary{arrows,calc}
\usepackage{verbatim}
\usetikzlibrary{decorations.pathreplacing,decorations.pathmorphing}
\usepackage[numbers,sort&compress]{natbib}
\usepackage{dsfont}

\numberwithin{equation}{section}
\newtheorem{theorem}{Theorem}[section]
\newtheorem{lemma}[theorem]{Lemma}
\newtheorem{definition}[theorem]{Definition}
\newtheorem{remark}[theorem]{Remark}

\newtheorem{proposition}[theorem]{Proposition}

 \reversemarginpar
 
 \newcommand{\Pf}{\text{Pf}}
\newcommand{\dq}{d_qx}
\newcommand{\al}{{\alpha}}
\newcommand{\be}{{\beta}}
\newcommand{\ab}{{(\alpha,\beta)}}

\newcommand{\z}{\mathbb{Z}}

\newcommand{\tih}{\tilde{h}}
 
\begin{document}

\title[On the discrete symplectic ensemble]{Classical discrete symplectic ensembles on the linear and exponential lattice: skew orthogonal polynomials and correlation functions} 

\author{Peter J. Forrester}
\address{School of Mathematical and Statistics, ARC Centre of Excellence for Mathematical and Statistical Frontiers, The University of Melbourne, Victoria 3010, Australia}
\email{pjforr@unimelb.edu.au}

\author{Shi-Hao Li}
\address{ School of Mathematical and Statistics, ARC Centre of Excellence for Mathematical and Statistical Frontiers, The University of Melbourne, Victoria 3010, Australia}
\email{shihao.li@unimelb.edu.au}

\subjclass[2010]{}
\date{}

\dedicatory{}


\begin{abstract}
The eigenvalue probability density function for symplectic invariant random matrix ensembles can be generalised to
discrete settings involving either a linear or exponential lattice. The corresponding correlation functions can be expressed in terms of
certain discrete, and $q$, skew orthogonal polynomials respectively. We give a theory of both of these classes of polynomials, and the correlation kernels
determining the correlation functions, in the cases that the weights for the corresponding discrete unitary ensembles are classical.
Crucial for this are certain difference operators which relate the relevant symmetric inner products to the skew
symmetric ones, and have a tridiagonal action on the corresponding (discrete or $q$) orthogonal polynomials.
\end{abstract}

\maketitle
\section{Introduction}
\subsection{Continuous invariant ensembles}
In the theory of random matrices (see e.g.~\cite[Ch.~5]{forrester10}) an ensemble of Hermitian matrices is said to have a unitary symmetry if its eigenvalue probability
density (PDF) is of the form
\begin{equation}\label{1.1}
{1 \over Z_N} \prod_{l=1}^N w(x_l) \prod_{1 \le j < k \le N} (x_k - x_j)^2.
\end{equation}
For example, choosing Hermitian matrices according to a Gaussian weight proportional to $e^{- {\rm Tr} \, X^2}$ specifies a unitary invariant ensemble --- known as the
Gaussian unitary ensemble --- with
$w(x)$ in (\ref{1.1}) equal to $e^{- x^2}$.

For general non-negative $w(x)$ --- referred to as a weight --- the $k$-point correlation $\rho_{N,k}(x_1,\dots, x_k)$ is specified by integrating (\ref{1.1}) over $x_{n+1}, \dots,
x_N$, and multiplying by $N(N-1) \cdots (N-n+1)$ as a normalisation. It is a standard result that
\begin{equation}\label{1.2}
\rho_{N,k}(x_1,\dots, x_k) =  \det \Big [ \hat{K}_N(x_{j_1}, x_{j_2})\Big]_{j_1,j_2 = 1}^k,
\end{equation}
for a certain kernel function $\hat{K}_N(x,y)$ independent of $k$. The structure (\ref{1.2}) specifies the eigenvalues of Hermitian matrices with a unitary symmetry as
examples of determinantal point processes; see e.g.~\cite{Bo11}.

Significant too is the precise functional form of $\hat{K}_N(x,y)$. Denote by $\{p_l(x) \}_{l=0,1,\dots}$ the set of monic polynomials orthogonal with respect to the weight
$w(x)$,
\begin{equation*}
\int_I w(x) p_j(x) p_k(x) \, dx = h_j \delta_{j,k}.
\end{equation*}
Here $h_j > 0$ denotes the normalisation, and $I$ denotes the interval of support of $w(x)$. One then has (see e.g.~\cite[Props.~5.1.1 and 5.1.2]{forrester10})
\begin{align}\label{1.3a}
\begin{aligned}
\hat{K}_N(x,y) &= \Big ( w(x) w(y) \Big )^{1/2} \sum_{l=0}^{N-1} {p_l(x) p_l(y) \over h_l}\\
& =  { ( w(x) w(y)  )^{1/2} \over h_{N-1}}
{p_N(x) p_{N-1}(y) - p_{N-1}(x) p_N(y) \over x - y},
\end{aligned}
\end{align}
where the identity implied by the final equality is known as the Christoffel-Darboux formula. 

In classical random matrix theory the generalisation of (\ref{1.1}) to the form
\begin{equation}\label{1.4}
{1 \over Z_{N, \beta}} \prod_{l=1}^N w_\beta(x_l) \prod_{1 \le j < k \le N} |x_k - x_j|^\beta
\end{equation}
for $\beta = 1$ and 4 is also prominent. Thus (\ref{1.4}) with $\beta = 1$ occurs as the eigenvalue PDF for real symmetric matrices with orthogonal symmetry.
With $\beta = 4$ it occurs as the eigenvalue PDF for $2N \times 2N$ Hermitian matrices with entries of the block form
\begin{equation}\label{1.5}
\begin{bmatrix} z & w \\
- \bar{w} & \bar{z} \end{bmatrix},
\end{equation}
the latter being a $2 \times 2$ matrix representation of quaternions, assuming too an invariance of the distribution with respect to conjugation by unitary symplectic matrices.
Both these cases are examples of Pfaffian point processes, with the general $k$-point correlation function having the form
\begin{equation*}
\rho_{N,k}(x_1,\dots, x_k) = {\rm Pf} \, \Big [ A_{N,\beta}(x_{j_1}, x_{j_2}) \Big ]_{j_1, j_2 =1,\dots, k}
\end{equation*}
for a certain $2 \times 2$  matrix $A_{N,\beta}(x,y)$, anti-symmetric with respect to interchange of $x$ and $y$, and independent of $k$.
The latter has the particular structure
\begin{equation*}
A_{N,\beta}(x,y) = \begin{bmatrix} I_{N,\beta}(x,y) & T_{N,\beta}(x,y) \\
- T_{N,\beta}(y,x) & D_{N,\beta}(x,y) \end{bmatrix}
\end{equation*}
with $I_{N,\beta}$, $D_{N,\beta}$ related to $T_{N,\beta}$ according to (see e.g.~\cite[Ch.~6; there $T_{N,\beta}(x,y)$ is denoted by
$S_4(x,y)$ (for $\beta = 4$) and by $S_1(x,y)$ for $\beta = 1$]{forrester10})
\begin{equation*}
I_{N,\beta}(x,y) = \int_x^y T_{N,\beta}(x,y') \, dy', \qquad
D_{N,\beta}(x,y) = {\partial \over \partial x} T_{N,\beta}(x,y).
\end{equation*}

While the PDFs (\ref{1.1}) and (\ref{1.4}) relate to continuous variables, there are  prominent examples from the setting
of combinatorial/ integrable probability (see e.g.~\cite{BP14}, \cite{FMS11}, \cite{FNR06},  \cite{FR02a}, \cite{FR02b}, \cite{Jo01}, \cite{Jo02}, \cite{LW16}) that give rise to PDFs of the same or an analogous form, but with
the variables taking on discrete values, typically from $\{k\}_{k=-\infty}^\infty$ (linear lattice) or
$\{q^k \}_{k=-\infty}^\infty$ (exponential lattice). Our interest in this paper is to identify special inter-relations that hold in
such discrete settings between analogues of ensembles with unitary and symplectic symmetry. To be more explicit, further
theory from the continuous case is required. 

\subsection{Inter-relations between continuous invariant ensembles with unitary and symplectic symmetry}
In the case of unitary symmetry we know from (\ref{1.3a}) that the correlation kernel can be expressed in terms of 
orthogonal polynomials corresponding to the weight $w(x)$ in (\ref{1.1}). It is similarly true that the quantity $S_{N,\beta}(x,y)$
determining the kernel matrix $A_{N,\beta}(x,y)$ for $\beta = 1$ and 4 can be expressed in terms of certain polynomials
associated with a skew inner product. The details are different depending on whether $\beta = 1$ and $\beta = 4$. Our interest
in the present paper is the case $\beta = 4$, when the relevant skew inner product reads
\begin{equation*}
\langle f, g \rangle_4 := \int_I w_4(x) \Big ( f(x) g'(x) - f'(x) g(x) \Big ) \, dx.
\end{equation*}
The polynomials of interest, $\{Q_j(x) \}_{j=0,1,\dots}$ say, are required to have the skew orthogonality property
\begin{equation}\label{1.9a}
\langle Q_{2m}, Q_{2n+1} \rangle_4 = q_m \delta_{m,n}, \qquad
\langle Q_{2m}, Q_{2n} \rangle_4 =\langle Q_{2m+1}, Q_{2n+1} \rangle_4  = 0.
\end{equation}
It is easy to see that skew orthogonality property holds if we make the replacement
 \begin{equation}\label{1.10}
 Q_{2m+1}(x) \mapsto  Q_{2m+1}(x)  + \gamma_{2m} Q_{2m}(x)
 \end{equation}
 for arbitrary $\gamma_{2m}$; in practice $\gamma_{2m}$ is chosen for convenience. In terms of these polynomials,
 one has (see e.g.~\cite[Prop.~6.1.6]{forrester10})
 \begin{equation*}
 T_4(x,y) = \sum_{m=0}^{N-1} {(w_4(x))^{1/2} \over q_m}
 \bigg ( Q_{2m}(x) {d \over d y} \Big (  (w_4(y))^{1/2}  Q_{2m+1}(y) \Big ) -
  Q_{2m+1}(x) {d \over d y} \Big (  (w_4(y))^{1/2}  Q_{2m}(y) \Big ) \bigg ).
 \end{equation*}
 
 Interplay between the above formulas holds if we first choose $w(x)$ as one of the so-called
 classical weights
  \begin{equation*} 
  w(x) = \left \{
  \begin{array}{ll} e^{-x^2}, & {\rm Hermite} \\
  x^a e^{-x} \: (x > 0), & {\rm Laguerre} \\
  (1 - x)^a (1+x)^b \: (-1 < x < 1), & {\rm Jacobi} \\
  (1 + x^2)^{- \alpha}, & {\rm Cauchy}.
  \end{array} \right.
  \end{equation*}
  These weights are distinguished by their logarithmic derivative being
 expressible as a rational function, 
 \begin{equation}\label{1.12a}  
 {w'(x) \over w(x)} = - {g(x) \over f(x)},
  \end{equation}
  with the degree of $f$ less
  than or equal to 2, and the degree of $g$ less than or equal to 1. Explicitly
 \begin{equation*}
  (f,g) = \left \{
  \begin{array}{ll} (1,2x), & {\rm Hermite} \\
  (x,x-a), & {\rm Laguerre} \\
  (1 - x^2, (a-b) + (a+b) x), & {\rm Jacobi} \\
  (1 + x^2, 2 \alpha x), & {\rm Cauchy}.
  \end{array} \right.
  \end{equation*}  
 It is a consequence of the low degrees of $f$ and $g$ that with
 \begin{equation}\label{1.12c}   
 \mathcal A := f {d \over dx}  +  {f' - g \over 2},
 \end{equation}
 and with $\{ p_k(x) \}$ the corresponding  set of monic orthogonal polynomials,
 \begin{equation}\label{1.12d}  
  \mathcal A  p_k(x) = - {c_k \over h_{k+1}} p_{k+1}(x)  + {c_{k-1} \over h_{k-1}} p_{k-1}(x),
 \end{equation}    
 for certain (easily determined) constants $\{c_k\}$; see \cite{adler00}.
  
   With one of the classical  forms of $w(x)$ assumed, choose
    \begin{equation}\label{1.13} 
 w_4(x) = f(x) w(x).
\end{equation} 
  It is shown in \cite{adler00} that then the corresponding skew orthogonal system satisfying (\ref{1.9a})
  can be expressed in terms of the monic orthogonal polynomials corresponding to $w(x)$, $\{p_k(x)\}$,
  and the constants $\{c_k\}$ in (\ref{1.12d}) according to
    \begin{align}\label{1.13a} 
Q_{2j+1}(x) & = p_{2j+1}(x), \nonumber \\
Q_{2j}(x) & = \Big ( \prod_{p=0}^{j-1} {c_{2p+1} \over c_{2p}} \Big )
\sum_{l=0}^j \prod_{p=0}^{l-1} {c_{2p} \over c_{2p+1}} p_{2l}(x) \nonumber \\
q_m & = c_{2m}. 
\end{align}
It is also shown in   \cite{adler00} that  
  \begin{equation}\label{1.13b} 
  T_4(x,y)  = {1 \over 2} \Big ( {f(x) \over f(y)} \Big )^{1/2} \bigg ( K_{2N}(x,y) +   (w(x) w(y))^{1/2} {c_{2N - 1} \over c_{2N}} 
{p_{2N}(y) \over h_{2N}} Q_{2N-2}(x) \bigg ).
\end{equation}  
Our primary aim in this paper is to introduce the notion of classical discrete weights in relation to discretisations of (\ref{1.1})
and (\ref{1.5}) --- the latter restricted to $\beta = 4$ --- on linear and exponential lattices, and to derive formulas analogous to (\ref{1.12c}), (\ref{1.12d}), (\ref{1.13a}) and (\ref{1.13b}) in this setting.

The notion of a classical weight relies on identifying a Pearson-type equation. This is given by (\ref{dp}) for the linear lattice, and by (\ref{qp})  for the exponential lattice.
The analogue of the operator (\ref{1.12c}) is given by (\ref{oa}) and (\ref{qa}) respectively, with action on the  corresponding family of orthogonal polynomials given by
(\ref{1.12e}) and (\ref{4.24a}). We give the analogue of (\ref{1.13a}) for the linear and exponential lattice in Propositions \ref{P3.9} and \ref{P4.6},
and that of (\ref{1.13b}) in Propositions \ref{P3.9} and \ref{P4.7}.

\section{Preliminary: Some facts about discrete orthogonal polynomials}\label{pre}
In the first part of this section we collect together some basic facts about the classical discrete orthogonal polynomials, as required for later development.
One can refer to \cite{nikiforov86} for more details. We conclude the section by giving the explicit form of the correlation functions for the
discretisations of (\ref{1.1}) on the linear and exponential lattices.

\begin{definition}\label{D2.1}
Let $x(t): \mathbb R \to \mathbb R$ be a monotonic function of $t$. The values $x_i := x(i)$, $i \in \mathbb Z$ are said to define
lattice points. Consider a function $\rho(x)$ --- referred to as a weight function --- which has the property of being non-negative at
all lattice points and permits finite moments. For general $h(x) = h(x(t))$ define
\begin{align*}
\Delta h(x(t)) & = h(x(t+1)) - h(x(t)) \\
\nabla h(x(t)) & = h(x(t)) - h(x(t-1)).
\end{align*}
A set of monic polynomials $\{p_n(x)\}_{n=0}^\infty$, i.e.~each $p_n(x)$ of degree $n$ with coefficient of $x^n$ unity, is said to be
orthogonal with respect to the weight function $\rho(x)$ if for each $n,m$
\begin{align}\label{orthogonality}
\sum_{i\in\mathbb{Z}}p_m(x_i)p_n(x_i)\rho(x_i)\Delta x_{i-1/2}=h_n\delta_{n,m}.
\end{align}
Here $h_n$ is the  normalisation with the property that  $h_n > 0$ for $\{x_i\}$ non-decreasing and $h_n < 0$ for $\{x_i\}$ non-increasing.
\end{definition}

\begin{proposition}
Suppose there are polynomials $f(x), g(x)$ such that
\begin{align}\label{pearson1}
\Delta[f(x_i)\rho(x_i)]=g(x_i)\rho(x_i)\Delta x(i-\frac{1}{2}), 
\end{align}
and define
\begin{align}\label{pearson2}
 \rho_n(x_i)=\rho(x_{i+n})\prod_{l=1}^n f(x_{i+l}).
 \end{align}
 Then, for suitable constants $B_n$, the monic orthogonal polynomials can be written as
 \begin{align*}
p_n(x_i)=\frac{B_n}{\rho(x_i)}\left(\frac{\nabla}{\nabla{x_{i+1/2}}}\cdots\frac{\nabla}{\nabla x_{i+n/2}}\right)\rho_n(x_i).
\end{align*}
\end{proposition}

\begin{remark}
By analogy with the continuous case, (\ref{pearson1}) is referred to a Pearson-type equation for $\rho(x)$, while
(\ref{pearson2}) is referred to as a Rodrigues-type equation for $p_n(x_i)$.
\end{remark}

There are two particular classes of lattices of interest in our study.

(1) The linear lattice $x(i)  = i$. In this case the orthogonality condition (\ref{orthogonality}) becomes
\begin{align}\label{do}
\sum_{x\in\mathbb{Z}} p_m(x)p_n(x)\rho(x)=h_n\delta_{n,m}.
\end{align}
At the same time, one can write the Pearson-type equation \eqref{pearson1} as
\begin{align}\label{dp}
\frac{\rho(x+1)}{\rho(x)}=\frac{f(x)+g(x)}{f(x+1)}.
\end{align}
Examples include the Hahn, Meixner, Krawtchouk and Charlier discrete orthogonal polynomial systems.

(2) The exponential lattice $x(i) = q^i$. Note that for $0 < q < 1$ the lattice points form a decreasing sequence.
Now  the orthogonality relation \eqref{orthogonality} reads
\begin{align*}
\sum_{s\in\mathbb{Z}}p_m(q^s)p_n(q^s)\rho(q^s)q^{s-\frac{1}{2}}(q-1)=h_n\delta_{n,m}.
\end{align*}
According to the definition of Jackson's $q$-integral (see e.g.~\cite{ismail05})
\begin{align}\label{J}
\int_0^\infty f(x)d_qx=(1-q)\sum_{s=-\infty}^\infty f(q^s)q^s,
\end{align}
showing the orthogonality relation is equivalent to
\begin{align}\label{qo}
\int_0^\infty p_m(x)p_n(x)\rho(x)d_qx=\tilde{h}_n\delta_{n,m}\quad \text{with\quad $\tilde{h}_n=-q^{1/2}h_n$.}
\end{align}
Furthermore, the Pearson-type equation in this case reads
\begin{align}\label{qp}
\frac{\rho(qx)}{\rho(x)}=\frac{f(x)-q^{-\frac{1}{2}}(1-q)xg(x)}{f(qx)}.
\end{align}
Examples include the $q$-analogue of the Hahn, Meixner,  Krawtchouk  and Charlier polynomial systems, and their degeneration cases, like the Al-Salam $\&$ Carlitz polynomials, little $q$-Jacobi polynomials and so on  \cite{koekoek96}.

Associated with the orthogonalities (\ref{do}) and (\ref{qo}) are the symmetric inner products $\langle\cdot,\cdot\rangle  : \, \mathbb{R}[x]\times\mathbb{R}[x]\to\mathbb{R}$
specified by
\begin{align}\label{oip}
\langle \phi(x),\psi(x)\rangle:= \left \{
\begin{array}{ll}
\sum_{i\in\mathbb{Z}}\phi(x_i)\psi(x_i)\rho(x_i), & {\rm linear \: lattice} \\
\int_0^\infty \phi(x) \psi(x) \rho(x) \, d_qx, & {\rm exponential \: lattice}. \end{array} \right.
\end{align}

It is of importance to point out that all of the orthogonal polynomials $\{p_n(x)\}_{n=0}^\infty$ of interest
span a Hilbert space $\mathcal{H}$. Because of this, the projection operator from $\mathcal{H}\to\mathcal{H}$
with respect to the inner product (\ref{oip})
\begin{align}\label{delta}
\delta(x,y) := \sum_{m=0}^{\infty}\frac{1}{\hat{h}_m}p_m(x)p_m(y),
\end{align}
where $\hat{h}_m = h_m$ as in (\ref{do}) (linear case) and $\hat{h}_m = \tilde{h}_m$ (exponential case) has a reproducing property.

\begin{proposition}\label{reproducing}
For general $\xi\in\mathcal{H}$, we have 
\begin{align}\label{ddelta}
\langle \delta(x,y),\xi(x)\rangle=\xi(y).
\end{align}
\end{proposition}
\begin{proof}
Since $\{p_n(x)\}_{n=0}^\infty$ form a basis for the Hilbert space $\mathcal{H}$, for $\xi\in\mathcal{H}$, we can write $\xi(x)=\sum_{i=0}^\infty a_ip_i(x)$ with constants $a_i\in\mathbb{R}$. Then from the orthogonality relation, we have
\begin{align*}
\langle \delta(x,y), \xi(x)\rangle=\sum_{m=0}^\infty\sum_{i=0}^\infty \frac{1}{h_m}a_ip_m(y)\langle p_m(x),p_i(x)\rangle=\sum_{i=0}^\infty a_ip_i(y)=\xi(y),
\end{align*}
as required.
\end{proof}


Let $\mathcal{H}_N$ be the subspace of $\mathcal{H}$ spanned by $\{p_n(x)\}_{n=0}^{N-1}$, and let  $K_N$ denote  the projection operator from $\mathcal{H}\to\mathcal{H}_N$
with respect to the inner product (\ref{oip}). Its kernel
\begin{align}\label{cd}
K_N(x,y)=\sum_{m=0}^{N-1}\frac{1}{\hat{h}_m}p_m(x)p_m(y)
\end{align}
is formally the same as the polynomial part of the correlation kernel (\ref{1.3a}). Since it too can be summed according to the
Christoffel-Darboux summation formula, it will be referred to as the Christoffel-Darboux kernel of the discrete unitary ensemble.
As in the continuous case, this kernel specifies the general $k$-point correlation function of the ensemble. 

Explicitly,  in the case of the linear lattice the discrete unitary ensemble is
defined by the PDF (\ref{1.1}), with each $x_l = n_l$, $ n_l\in \mathbb Z$.  The same working
as leading to (\ref{1.2}) and (\ref{1.3a}) (see \cite[Props.~5.1.1 and 5.1.2]{forrester10}) gives for the 
correlation function
\begin{equation}\label{M1a}
\rho_{N,k}(n_1,\dots,n_k) = \prod_{l=1}^k w(n_l) \, \det \Big [ K_N(n_{j_1}, n_{j_2}) \Big ]_{j_1, j_2 = 1,\dots, k}.
\end{equation}
The joint probability distribution for the discrete ensemble in the case of the exponential lattice is most conveniently written as a measure rather than
a PDF,
\begin{align*}
{1 \over Z_N} \prod_{1 \le j < k \le N} (x_k - x_j)^2  \prod_{l=1}^N w(x_l) d_qx_l,
\end{align*}
where each $x_l  = q^{n_l}$ for some $n_l \in \mathbb Z$. The corresponding correlation functions then relate to the
Christoffel-Darboux kernel by
\begin{equation*}
\rho_{N,k}(n_1,\dots,n_k) = \prod_{l=1}^k  w(q^{n_l}) \, \det \Big [ K_N(q^{n_{j_1}}, q^{n_{j_2}}) \Big ]_{j_1, j_2 = 1,\dots, k}.
\end{equation*}

\section{The discrete symplectic ensemble on a linear lattice}
In this section, we consider the discrete symplectic ensemble on a linear lattice, with our aim being to develop a theory
of the classical cases analogous to that done in \cite{adler00} for the symplectic ensemble with continuous weights.
 We remark that a detailed study of the  discrete symplectic ensemble on linear lattice has previously been given in \cite{borodin09},
 but from a different point of view.
For completeness, we first consider general weights, before specialising to the classical cases.

\begin{definition}
Let $x_1<\cdots<x_N$ be an ordered set of distinct integers.
The joint probability density function of the discrete symplectic ensemble on a linear lattice with weight function $\omega(x)$ is specified by
\begin{align}\label{3.1a}
\frac{1}{Z_{N}}\prod_{1\leq i<j\leq N} (x_i-x_j)^2(x_i-x_j-1)(x_i-x_j+1)\prod_{i=1}^N \omega(x_i),
\end{align}
where $Z_{N}$ is the partition function, assumed to be finite, and given by
\begin{align*}
Z_{N}=\sum_{\xi_N}\prod_{1\leq i<j\leq N} (x_i-x_j)^2(x_i-x_j-1)(x_i-x_j+1)\prod_{i=1}^N \omega(x_i),
\end{align*}
where $\xi_N$ is a configuration space in $\mathbb{Z}^N$ defined by
\begin{align*}
\xi_N=\{(x_1,\cdots,x_N)|x_1<\cdots<x_N,\, x_i\in\mathbb{Z}\}.
\end{align*}
\end{definition}
According to the \cite[Lemma 5.1 and Lemma 5.2]{borodin09}, one can see the partition function $Z_N$ can be written as a Pfaffian 
\begin{align}\label{borodin}
Z_N=\Pf[A_{i,j}]_{i,j=0}^{2N-1}, \quad A_{i,j}=\sum_{x\in\mathbb{Z}}\left[\pi_i(x)\pi_j(x+1)-\pi_i(x+1)\pi_j(x)\right]\omega(x),
\end{align}
where $\pi_i(x)$ is the monic polynomials of order $i$, $i=0,\cdots,2n-1$. Now we proceed to use skew tridiagonalisation to show that a family of discrete skew-orthogonal polynomials are inherent in this model.

\subsection{Discrete skew orthogonal polynomials}\label{dsop1}
In this subsection, we are going to show how to relate the discrete symplectic ensemble on the integer lattice to the discrete skew orthogonal polynomials.
\begin{definition}\label{D3.2}
Consider a skew-symmetric inner product $\langle\cdot,\cdot\rangle_{s,\omega}$ with weight function $\omega(x)$, defined on $\mathbb{R}[x]\times\mathbb{R}[x]\to\mathbb{R}$, admitting the form
\begin{align}\label{si}
\begin{aligned}
\langle \phi(x),\psi(x)\rangle_{s,\omega}&=\sum_{x\in\mathbb{Z}}\left[\phi(x)\psi(x+1)-\phi(x+1)\psi(x)\right]\omega(x)\\&=\sum_{x\in\z}[\phi(x)\Delta\psi(x)-\psi(x)\Delta\phi(x)]\omega(x)
\end{aligned}
\end{align}
with $\Delta\phi(x)=\phi(x+1)-\phi(x)$ for arbitrary function $\phi(x)\in\mathbb{R}[x]$ (recall Definition \ref{D2.1}).
\end{definition}

From this definition, a moment matrix $\left(m_{i,j}\right)_{i,j\geq0}$ could be generated by the skew-symmetric inner product $\langle x^i,x^j\rangle_{s,\omega}$, and one sees $A_{i,j}$ in \eqref{borodin} can then be chosen as $m_{i,j}$.
Next we introduce a family of discrete \textbf{monic} skew-orthogonal polynomials $\{Q_k(x)\}_{k=0}^\infty$, from which the moment matrix can be skew tridiagonalised.

\begin{definition}
Let $\langle \cdot , \cdot \rangle_{s, \omega}$ be specified as in Definition \ref{D3.2}.
Analogous to (\ref{1.9a}) in the continuous case, the discrete skew-orthogonal polynomials, unique up to the mapping
(\ref{1.10}), are specified by the  relations
\begin{align}\label{dso}
\begin{aligned}
&\langle Q_{2n}(x),Q_{2m+1}(x)\rangle_{s,\omega}=-\langle Q_{2m+1}(x),Q_{2n}(x)\rangle_{s,\omega}=u_n\delta_{n,m},\\
&\langle Q_{2m}(x), Q_{2n}(x)\rangle_{s,\omega}=\langle Q_{2m+1}(x),Q_{2n+1}(x)\rangle_{s,\omega}=0.
\end{aligned}
\end{align}
\end{definition}

Note that the skew-orthogonal relation \eqref{dso} is equivalent to the condition
\begin{align*}
\langle Q_{2n}(x),x^i\rangle_{s,\omega}=u_n\delta_{i,2n+1},\quad \langle Q_{2n+1}(x),x^i\rangle_{s,\omega}=-u_n\delta_{i,2n+1},\quad 0\leq i\leq 2n+1.
\end{align*}
Moreover, if we denote
\begin{align*}
Q_{2n}(x)=\sum_{k=0}^{2n}a_{2n,k} x^k, \quad Q_{2n+1}(x)=\sum_{k=0}^{2n+1}b_{2n+1,k}x^{k}
\end{align*}
with $a_{2n,2n}=b_{2n+1,2n+1}=1$, then the relations \eqref{dso} are linear systems for solving $a_{2n,k}$ and $b_{2n+1,k}$. One can find a detailed computation in \cite[Section 2.1]{chang18}, where it is shown that after solving the linear system and employing a determinant identity, the (discrete) skew-orthogonal polynomials can be written in terms of Pfaffians
\begin{align*}
Q_{2n}(x)=\frac{1}{\tau_{2n}}\Pf(0,\cdots,2n,x),\quad Q_{2n+1}(x)=\frac{1}{\tau_{2n}}\Pf(0,\cdots,2n-1,2n+1,x)
\end{align*}
with the elements
\begin{align*}
\Pf(i,j)=m_{i,j}=\langle x^i,x^j\rangle_{s,\omega}, \quad \Pf(i,x)=x^i, \quad \tau_{2n}=\Pf(m_{i,j})_{i,j=0}^{2n-1}.
\end{align*}
Here the arbitrary constant $\gamma_{2n}$ in (\ref{1.10}) has been chosen so that the coefficient of $x^{2n}$ in $Q_{2n+1}(x)$  vanishes.

One important quantity here is the $\tau_{2n}$, which is the normalisation factor of the skew orthogonal polynomials. It is also the partition function $Z_N$ in \eqref{borodin} if one takes $N=n$. Moreover, one can compute that
\begin{align*}
\langle Q_{2n}(x),x^{2n+1}\rangle&=\frac{1}{\tau_{2n}}\langle \Pf(0,\cdots,2n,x),x^{2n+1}\rangle_{s,\omega}=\frac{1}{\tau_{2n}}\sum_{i=0}^{2n}(-1)^i\Pf(0,\cdots,\hat{i},\cdots,2n)\langle x^i,x^{2n+1}\rangle_{s,\omega}\\&=\frac{1}{\tau_{2n}}\sum_{i=0}^{2n}(-1)^i\Pf(0,\cdots,\hat{i},\cdots,2n)\Pf(i,2n+1)=\frac{\tau_{2n+2}}{\tau_{2n}}=u_n.
\end{align*}
And upon skew tridiagonalising the moment matrix, it follows
\begin{align}\label{3.3a}
\Pf[\langle x^i,x^j\rangle_{s,\omega}]_{i,j=0}^{2n-1}=\Pf[\langle Q_i(x),Q_j(x)\rangle_{s,\omega}]_{i,j=0}^{2n-1}=\prod_{i=0}^{n-1}u_i=\tau_{2n}.
\end{align}

\subsubsection{Skew orthogonal polynomials revisited---moment matrix realisation} In this section
we focus attention to the skew-symmetric moment matrix $M_\infty=(m_{i,j})_{i,j=0}^\infty$. We always assume the even-order minor principle of this moment matrix is nonzero, and thus
$\tau_{2n} \ne 0$ for each non-negative integer $n$.

Since $M_\infty$ is a skew-symmetric matrix, one can apply a generalised LU decomposition \cite{bunch82}, or so-called skew Borel decomposition \cite{adler99}, to the moment matrix $M_\infty$ such that
\begin{align}\label{sbd}
M=S^{-1}JS^{-\top},
\end{align}
where $S$ is a lower triangular matrix with diagonals $1$ and $J$ is a block matrix with $2\times 2$ block matrix, admitting the form
\begin{align*}
J=\left(
\begin{array}{ccccc}
0&u_0&&&\\
-u_0&0&&&\\
&&0&u_1&\\
&&-u_1&0&\\
&&&&\ddots
\end{array}
\right).
\end{align*}
Therefore, if we take $SM_\infty S^\top$, then this moment matrix is skew tridiagonalised. This has been used in (\ref{3.3a}).

On the other hand, by denoting $\chi(x)=(1,x,x^2,\cdots)^\top$, we can define a family of polynomials $\{Q_n(x)\}_{n=0}^\infty$ by
\begin{align*}
Q_n(x)=(S\chi(x))_n,
\end{align*}
where $(a)_n$ means the $n$-th component of the vector $a$. It is clear that $Q_n(x)$ is a monic polynomial of order $n$ since $S$ is a  lower triangular matrix with diagonals $1$. Moreover, the polynomials $\{Q_{n}(x)\}_{n=0}^\infty$ are skew orthogonal under the skew-symmetric inner product \eqref{si} so that
\begin{align*}
\langle S\chi(x),\chi(x)^\top S^\top\rangle_{s,\omega}=S\langle \chi(x),\chi(x)^\top\rangle_{s,\omega}S^\top=SM_\infty S^\top=J,
\end{align*}
which corresponds to the definition \eqref{dso}.

\subsection{Christoffel-Darboux kernel and Pfaffian point process}\label{dsop2}
We know from  \cite{borodin09} that the general $k$-point correlation function $\rho_{N,k}$ for the discrete symplectic ensemble (\ref{3.1a}) can be written as
a Pfaffian. However no use was made of the skew orthogonal polynomials. It is our objective here to develop the associated theory with the
skew orthogonal polynomials regarded as central.

The $k$-point correlation function corresponding to (\ref{3.1a}) is specified by the multi-dimensional summation
\begin{align*}
\rho_{N,k}(x_1,\cdots,x_k)=\frac{1}{\tau_{2N}}\sum_{(x_{k+1}<\cdots<x_N)\subset \mathbb{Z}^{N-k}} \prod_{1\leq i<j\leq N}(x_i-x_j)^2(x_i-x_j-1)(x_i-x_j+1)\prod_{i=1}^N \omega(x_i).
\end{align*}
As for the correlations of the continuous symplectic symmetry ensemble (\ref{1.5}) with $\beta = 4$, 
 $\rho_{N,k}$ can be expressed as a Pfaffian, and is fully determined by a particular $2 \times 2$ anti-symmetric matrix.

\begin{proposition}
Let  $M$ denote the moment matrix $M=(m_{i,j})_{i,j=0}^{2N-1}$, and define $\Delta x^i=(x+1)^i-x^i$.
The correlation function $\rho_{N,k}$ is a Pfaffian point process, having correlations of the form $\rho_{N,k}= \prod_{l=1}^n \omega(x_l)
\Pf\left[\tilde{S}_N(x_{j_1},x_{j_2})\right]_{j_1, j_2=1}^k$, where $\tilde{S}_N(x,y)$ is a $2\times 2$ anti-symmetric matrix
\begin{align}\label{correlationkernel}
\tilde{S}_N(x,y)=\left(
\begin{array}{cc}
\sum_{0\leq i,j\leq 2N-1} x^iM_{i,j}^{-\top}y^j & \sum_{0\leq i,j\leq 2N-1} x^iM_{i,j}^{-\top} \Delta y^j\\
\sum_{0\leq i,j\leq 2N-1} \Delta x^iM^{-\top}_{i,j} y^j &\sum_{0\leq i,j\leq 2N-1}\Delta x^iM_{i,j}^{-\top}\Delta y^j
\end{array}
\right).
\end{align}
\end{proposition}

This proposition is essentially a special case of \cite[Corollary 1.3]{rains00}, if one takes the measure $\lambda$ therein as supported on the integer lattice,
and $\phi_i(x)=x^{i-1}$, $\psi_i(x)=\Delta \phi_i(x)=(x+1)^{i-1}-x^{i-1}$. The slight difference is that we have considered the ordered set $x_{k+1}<\cdots<x_N$ and therefore a factor of $(N-k)!$ is missed. Since the matrix $M$ is skew tridiagonalisable (which is equivalent to $\tau_{2n}\not=0$),  the inverse of $M$ exists. 

\begin{remark}\label{R3.5}
With respect to the formula for the correlations  \cite[Corollary 1.3]{rains00},
the authors of  \cite{borodin09} took $\phi_i(x) =  p_j(x)$ --- the discrete orthogonal polynomials $\{p_j(x)\}$ from  (\ref{do})  ---
and $\psi(x)=\Delta \phi(x)$. This choice provides one strategy to investigate the relation to Christoffel-Darboux kernel using the theory of \cite{widom99}.
However in this article we are viewing  the discrete skew orthogonal polynomials as an inherent structure of the discrete symplectic ensemble and therefore 
seek to analyse the kernel within this framework.
\end{remark}

Therefore, according to the skew Borel decomposition \eqref{sbd}, one can obtain
\begin{align}
\sum_{0\leq i,j\leq 2N-1} x^iM^{-\top}_{i,j}y^j&=-\chi^\top(x)S^\top J^{-1}S\chi(y)=\sum_{i=0}^{N-1}\frac{1}{u_i}\left(Q_{2i}(x)Q_{2i+1}(y)-Q_{2i}(y)Q_{2i+1}(x)\right)\nonumber\\
&:=S_N(x,y).\label{scd}
\end{align}
By anaolgy with (\ref{cd}) this will be referred to as the symplectic
Christoffel-Darboux kernel.  One can simplify the correlation kernel \eqref{correlationkernel} with the help of the symplectic Christoffel-Darboux kernel to obtain
\begin{align}\label{ss}
\tilde{S}_N(x,y)=\left(
\begin{array}{cc}
S_N(x,y)& \Delta_x S_N(x,y)\\
\Delta_y S_N(x,y)& \Delta_x\Delta_y S_N(x,y)
\end{array}
\right).
\end{align}
We remark that  the symplectic Christoffel-Darboux kernel admits the reproducing property
\begin{align*}
&\langle S_N(x,y),S_N(y,z)\rangle_{s,\omega}\\
&=\sum_{i,j=0}^{N-1}\frac{1}{u_iu_j}
[-Q_{2i}(x)Q_{2j+1}(z)\langle Q_{2i+1}(y), Q_{2j}(y)\rangle_{s,\omega}-Q_{2i+1}(x)Q_{2j}(z)\langle Q_{2i}(y), Q_{2j+1}(y)\rangle_{s,\omega}]\\
&=\sum_{i,j=0}^{N-1}\frac{1}{u_iu_j}[u_i\delta_{i,j}Q_{2i}(x)Q_{2j+1}(z)-u_i\delta_{i,j}Q_{2i+1}(x)Q_{2j}(z)]=S_N(x,z).
\end{align*}

In the continuous case, theory developed in \cite{adler00} gives a relationship between the symplectic Christoffel-Darboux kernel and the unitary Christoffel-Darboux kernel, in the
case of classical weights related by (\ref{1.13}). Our objective in the subsequent subsections is to give an analogous theory relating (\ref{scd}) to (\ref{cd}) for certain classical weights.

\subsection{Relationship between discrete skew orthogonal polynomials and discrete orthogonal polynomials}
To discovery the relationship between skew orthogonal polynomials and discrete orthogonal polynomials is an effective way to depict the linkage between the kernels and one can see \cite{adler00} and \cite[Chapter 6]{forrester10} for more details in continuous case. 

We know that in the continuous case a crucial role is played by the operator (\ref{1.12c}).
For the discrete case on a linear lattice, as an analogue define the operator
\begin{align}\label{oa}
\mathcal{A}_{\rm l}=g(x)T+f(x)(\Delta+\nabla),
\end{align}
where $g(x)$ and $f(x)$ are as in \eqref{dp}, $T$ is the shift operator defined by $T\phi(x)=\phi(x+1)$, and $\Delta \phi(x)=\phi(x+1)-\phi(x)$ and $\nabla\phi(x)=\phi(x)-\phi(x-1)$ are the same as in Definition \ref{D2.1}. This operator has the following key properties.

\begin{proposition}\label{P3.6}
Define the symmetric inner product
\begin{align*}
\langle \phi(x),\psi(x)\rangle=\sum_{x\in\z}\phi(x)\psi(x)\rho(x)
\end{align*}
as is consistent with the linear lattice case of (\ref{oip}). 
Define the skew symmetric inner product $\langle f(x),g(x)\rangle_{s,\omega}$ according to (\ref{si}).
Under the assumption that 
 $\rho(x)f(x)$ vanishing at the end points of support, the operator $\mathcal{A}_{\rm l}$ defined in \eqref{oa} satisfies
\begin{align}
& (1)\,\, \langle \mathcal{A}_{\rm l}\phi(x),\psi(x)\rangle=-\langle \phi(x),\mathcal{A}_{\rm l}\psi(x)\rangle;\label{ss}\\
& (2)\,\, \langle \phi(x),\mathcal{A}_{\rm l}\psi(x)\rangle=\langle \phi(x),\psi(x)\rangle_{s,f(x+1)\rho(x+1)}.\label{os}
\end{align}
\end{proposition}

\begin{proof}
Firstly, we prove the equality \eqref{ss}. Decompose
\begin{align*}
&\langle \mathcal{A}_{\rm l}\phi(x),\psi(x)\rangle+\langle \phi(x),\mathcal{A}_{\rm l}\psi(x)\rangle\\
&=\sum_{x\in\z}(\phi(x+1)\psi(x)+\phi(x)\psi(x+1))g(x)\rho(x)+\sum_{x\in\z}(\Delta\phi(x)\psi(x)+\phi(x)\Delta\psi(x))f(x)\rho(x)\\
&+\sum_{x\in\z}(\nabla\phi(x)\psi(x)+\phi(x)\nabla\psi(x))f(x)\rho(x):=A_1+A_2+A_3.
\end{align*}
Now from the discrete Pearson equation \eqref{dp}, we get
\begin{align*}
&A_2=-2\sum_{x\in\z}\phi(x)\psi(x)f(x)\rho(x)+\sum_{x\in\z}[\phi(x+1)\psi(x)+\phi(x)\psi(x+1)][\rho(x+1)f(x+1)-g(x)\rho(x)],\\
&A_3=2\sum_{x\in\z}\phi(x)\psi(x)f(x)\rho(x)-\sum_{x\in\z}[\phi(x-1)\psi(x)+\phi(x)\psi(x-1)]\rho(x)f(x).
\end{align*}
Moreover, since by assumption $\rho(x)f(x)$ vanishes at the end points of the supports, we see
\begin{align*}
A_2+A_3=-\sum_{x\in\z}[\phi(x+1)\psi(x)+\phi(x+1)\psi(x)]g(x)\rho(x)=-A_1,
\end{align*}
showing $\langle\mathcal{A}_{\rm l}\phi(x),\psi(x)\rangle+\langle \phi(x),\mathcal{A}_{\rm l}\psi(x)\rangle=0$.

Now we turn to \eqref{os}. This equation can be proved by noting
\begin{align*}
&\langle \phi(x),\mathcal{A}_{\rm l}\psi(x)\rangle\\
&=\sum_{x\in\z}\phi(x)\psi(x+1)\rho(x)g(x)+\sum_{x\in\z}\phi(x)\Delta\psi(x)f(x)\rho(x)+\sum_{x\in\z}\phi(x)\nabla\psi(x)f(x)\rho(x)\\
&=\sum_{x\in\z}\phi(x)\psi(x+1)f(x+1)\rho(x+1)-\sum_{x\in\z}\phi(x)\psi(x-1)f(x)\rho(x)\\
&=\sum_{x\in\z}[\phi(x)\psi(x+1)-\phi(x+1)\psi(x)]f(x+1)\rho(x+1)=\langle \phi(x),\psi(x)\rangle_{s,f(x+1)\rho(x+1)}.
\end{align*}
\end{proof}
\begin{remark}
In this proof, we have assumed $\rho(x)f(x)$ vanishes at the end points of support, which is always valid for the classical weights. The specific examples, such as Hahn polynomials supported on $[0,N]$, Charlier polynomials and Meixner polynomials supported on $[0,\infty)$ are demonstrated in Section \ref{example1}.
\end{remark}

Now we specialise to the classical cases, for which $f(x)$ and $g(x)$ in the discrete Pearson equation  \eqref{dp} are of order no bigger than $2$ and $1$ respectively. With this assumption $\mathcal{A}_{\rm l}$ maps polynomials of degree $n$ to polynomials of degree $n+1$ and thus
\begin{align*}
\mathcal{A}_{\rm l} p_n(x)=\sum_{j=0}^{n+1}a_{n,j}p_j(x),
\end{align*}
for some coefficients $a_{n,j}$.
Noting from the anti-commutativity property \eqref{ss} and discrete orthogonality \eqref{do} that
\begin{align*}
\langle \mathcal{A}_{\rm l} p_n,p_i\rangle=-\langle p_n,\mathcal{A}_{\rm l} p_i\rangle=0, \quad \text{if $i<n-1$},
\end{align*}
it follows $a_{n,i}=0$ for $i<n-1$. For $i=n-1$ and $i=n$, we have
\begin{align*}
&a_{n,n-1}h_{n-1}=\langle \mathcal{A}_{\rm l}p_n,p_{n-1}\rangle=-\langle p_n,\mathcal{A}_{\rm l}p_{n-1}\rangle=-a_{n-1,n}h_n,\\
&a_{n,n}h_n=\langle \mathcal{A}_{\rm l}p_n,p_n\rangle=-\langle p_n,\mathcal{A}_{\rm l} p_n\rangle=-a_{n,n}h_n, \: \: {\rm and \: thus} \: a_{n,n}=0,
\end{align*}
where $h_n$ is the normalisation constant in the discrete orthogonality \eqref{do}. 
From these equalities, one can conclude the following result.
\begin{proposition}
For the discrete orthogonal polynomials $\{p_n(x)\}_{n=0}^\infty$ and operator $\mathcal{A}_{\rm l}$ defined in 
\eqref{oa}, we have
\begin{align}\label{1.12e}
\mathcal{A}_{\rm l}p_n(x)=-\frac{c_n}{h_{n+1}}p_{n+1}(x)+\frac{c_{n-1}}{h_{n-1}}p_{n-1}(x),
\end{align}
where $c_n$ is a constant, depending on the degree of polynomials $p_n$ and the functions $f(x)$, $g(x)$ appearing in the discrete Pearson equation \eqref{dp}
(cf.~(\ref{1.12d})). Equivalently
\begin{align}\label{c}
\left[\langle p_j,\mathcal{A}_{\rm l}p_k \rangle \right]_{j,k=0}^{2N-1}=\mathbf{c}^\top\Lambda-\Lambda^\top\mathbf{c}:=\mathbf{C},
\end{align}
where $\mathbf{c}=(c_0,\cdots,c_{2N-1})^\top$ and $\Lambda$ is the shifted operator, whose entries on the leading upper diagonal are 
all $1$ and the other entries  are all $0$.
\end{proposition}

Since each $Q_j(x)$ and $p_j(x)$ is a monic polynomial of degree $j$ we can introduce a lower triangular matrix $\mathbf{T}$ with 1's on the diagonal such that \begin{align}\label{trans}
\left[ Q_j(x)\right]_{j=0}^{2N-1}=\mathbf{T}\left[ p_j(x)\right]_{j=0}^{2N-1}.
\end{align}
Moreover, the equations \eqref{os} and \eqref{trans} tell us 
\begin{align*}
\mathbf{U}:=\left[\langle Q_j,Q_k\rangle_{s,f(x+1)\rho(x+1)}\right]_{j,k=0}^{2N-1}=\left[\langle Q_j,\mathcal{A}Q_k\right]_{j,k=0}^{2N-1}=\mathbf{T}\left[\langle p_j,\mathcal{A}p_k\rangle\right]^{2N-1}_{j,k=0}\mathbf{T}^\top=\mathbf{T}\mathbf{C}\mathbf{T}^\top,
\end{align*}
and from the skew orthogonality \eqref{dso} we know $\mathbf{U}$ is a skew symmetric $2\times2$ block diagonal matrix, admitting the form
\begin{align*}
\mathbf{U}=\text{diag}\left(\left[\begin{array}{cc}0&u_0\\-u_0&0\end{array}\right],\cdots,\left[\begin{array}{cc}0&u_{N-1}\\-u_{N-1}&0\end{array}\right]\right).
\end{align*}
Since $\mathbf{U}$ and $\mathbf{C}$ are the tridiagonal matrices, and $\mathbf{T}$ is a lower triangular matrix, from the equation $\mathbf{T}^{-1}\mathbf{U}=\mathbf{C}\mathbf{T}^\top$, we know the left hand side is an upper Hessenberg matrix and the right one is a lower Hessenberg matrix, which means both sides of the equation are in fact tridiagonal matrices. If we set $\mathbf{T}^{-1}:=[t_{j,k}]_{j,k=0}^{2N-1}$ and consider the completely lower triangular part (not including the diagonal), to be denoted $( \cdot )_-$, one has
\begin{align*}
(\mathbf{T}^{-1}\mathbf{U})_-=\left[
\begin{array}{cccccc}
\star&&&&&\\
-u_0&\star&&&&\\
-u_0t_{21}&u_0t_{20}&\star&&&\\
-u_0t_{31}&u_0t_{30}&-u_1&\star&&\\
-u_0t_{41}&u_0t_{40}&-u_1t_{43}&u_1t_{42}&\star&\\
\vdots&\vdots&\vdots&\vdots&\vdots&\ddots
\end{array}
\right].
\end{align*}
On the other hand, from equation \eqref{c}, one has
\begin{align*}
(\mathbf{C}\mathbf{T}^\top)_-=\left(
(\mathbf{c}^\top\Lambda-\Lambda^\top \mathbf{c})\mathbf{T}^\top
\right)_-=\left(-\Lambda^\top\mathbf{c}\mathbf{T}^\top\right)_-=-\Lambda^\top\mathbf{c}.
\end{align*}
Equating these two formulas gives
\begin{align*}
&c_{2n}=u_n,\quad\quad n=0,\cdots,N-1,\\
&t_{2n+1,j}=0,\quad j=0,\cdots,2n-1,\\
&t_{2n,j}=0,\quad\quad j=0,\cdots,2n-3,2n-1,\\
&t_{2n,2n-2}=-\frac{c_{2n-1}}{c_{2n-2}},
\end{align*}
and therefore, one can get 
\begin{align}\label{pq}
\begin{aligned}
&p_{2n+1}(x)=Q_{2n+1}(x)+t_{2n+1,2n}Q_{2n}(x),\\
&p_{2n}(x)=Q_{2n}(x)-\frac{c_{2n-1}}{c_{2n-2}}Q_{2n-2}(x).
\end{aligned}
\end{align}

Solving this equation backwards and setting $t_{2n+1,2n}=0$ as permitted by the non-uniqueness of the skew-orthogonal polynomials $p_{2n+1}(x)$, we finally get
an expression for the discrete skew orthogonal polynomials in terms of the corresponding discrete orthogonal polynomials.

\begin{proposition}\label{P3.9}
Let $\{Q_j(x)\}$ be the skew orthogonal polynomials associated with (\ref{3.1a}), and choose for the weight $\omega(x) = f(x+1) \rho(x+1)$, where $\rho(x)$ is the original weight function of orthogonal polynomials $\{p_j(x)\}$ defined in \eqref{do} and 
$(f(x),g(x))$ satisfies the Pearson-type equation (\ref{pearson1}) with $f(x)$ of degree no bigger than 2, and $g(x)$ of degree no bigger than 1.
We have
\begin{align}\label{disstr}
Q_{2n+1}(x)=p_{2n+1}(x),\quad Q_{2n}(x)=\left(\prod_{j=0}^{n-1}\frac{c_{2j+1}}{c_{2j}}\right)\sum_{l=0}^{n}\prod_{j=0}^{l-1}\frac{c_{2j}}{c_{2j+1}}p_{2l}(x),
\end{align}
with normalisation $u_n = c_{2n}$.
\end{proposition}

\begin{remark}
The relations (\ref{disstr}) are formally the same as those for skew orthogonal polynomials the continuous symplectic invariant ensemble (\ref{1.4}) with
$\beta = 4$ and $w_4(x)$ given by (\ref{1.13}) \cite{adler00}.
\end{remark}

\subsubsection{Examples}\label{example1} 
In this part, Pearson pairs $(f,g)$ of Meixner, Charlier and Hahn polynomials are constructed. With the help of the Pearson pairs, we can obtain the coefficients $c_n$ in \eqref{1.12e} such that the discrete skew orthogonal polynomials with classical weight can be explicitly written down.
\begin{enumerate}
\item Meixner case. For the Meixner weight
\begin{align*}
\rho(x)=\frac{(\be)_x}{x!}a^x, \quad (\be)_x=\be\cdots(\be+x-1)
\end{align*}
with positive real number $\be$ and $0<a<1$, we can get the Pearson pair 
\begin{align*}
(f,g)=\left(x,(a-1)x+a\be\right).
\end{align*}
Moreover, the coefficients in \eqref{1.12e} can be written as
\begin{align*}
c_n=(1-a)h_{n+1},
\end{align*}
where $h_n$ is the normalisation constant of monic Meixner polynomials.

\item Charlier case. For the Charlier weight
\begin{align*}
\rho(x)=\frac{a^x}{x!}, \quad a>0
\end{align*}
supported on $[0,\infty)$, one can compute the Pearson pair
\begin{align*}
(f,g)=\left(x,a-x\right).
\end{align*}
One can show the coefficients  $c_n$ in \eqref{1.12e}  for the skew orthogonal Charlier polynomials are given by
\begin{align*}
c_n=h_{n+1},
\end{align*}
where $h_n$ is the normalisation constant of monic Charlier polynomials.

\item Hahn case. For the Hahn weight
\begin{align*}
\rho(x)={\al+x\choose x}{N+\be-x\choose N-x}, \quad \al>0,\,\be>0
\end{align*}
supported on $[0,N]$, one can compute the Pearson pair
\begin{align*}
(f,g)=\left(-x^2+(N+\be+1)x,-(\al+\be+2)x+N(\al+1)  \right).
\end{align*}
It is easy to determine the coefficients $c_n$ in \eqref{1.12e}  as
\begin{align*}
c_n=(n+\al+\be+2)h_{n+1},
\end{align*}
with $h_n$ the normalisation constant of monic Hahn polynomials.
\end{enumerate}
\subsection{Kernels between discrete unitary ensembles and discrete symplectic ensembles}
To express the discrete symplectic Christoffel-Darboux kernel \eqref{scd} in terms of its unitary counterpart \eqref{cd}, we need to use the relationship between the skew-symmetric inner product and symmetric inner product obtained in \eqref{os}.

Firstly, consider the skew symmetric inner product
\begin{align*}
\langle \phi(x),\psi(x)\rangle_{s,f(x+1)\rho(x+1)}=\sum_{x\in\mathbb{Z}}[\phi(x)\psi(x+1)-\phi(x+1)\psi(x)]f(x+1)\rho(x+1)
\end{align*}
and take $\phi(x)=\delta(x,y)$, where $\delta(x,y)$ is defined in \eqref{delta}. 
On the one hand, according to the equation \eqref{os} and Proposition \ref{reproducing}, one knows
\begin{align*}
\langle \delta(x,y),\psi(x)\rangle_{s,f(x+1)\rho(x+1)}=\langle \delta(x,y), \mathcal{A}_{\rm l} \psi(x)\rangle=\mathcal{A}_{\rm l}  \psi(y).
\end{align*}
On the other hand, from the definition of $\delta(x,y)$
\begin{align*}
\langle \delta(x,y),\psi(x)\rangle_{s,f(x+1)\rho(x+1)}=\sum_{n=0}^\infty\frac{p_n(y)}{h_n}\langle p_n(x),\psi(x)\rangle_{s,f(x+1)\rho(x+1)}.
\end{align*}
If we take $\psi(x)=Q_{2m}(x)$, combining these gives
\begin{align*}
\mathcal{A}_{\rm l}Q_{2m}(y)=\sum_{n=0}^\infty\frac{p_{n}(y)}{h_n}\langle p_n(x),Q_{2m}(x)\rangle_{s,\rho(x+1)f(x+1)}=-\frac{u_m}{h_{2m+1}}p_{2m+1}(y)
\end{align*}
where we used the equation \eqref{pq} for the final equality, and here $u_m$ is the normalisation factor defined in \eqref{dso}. Also, if we take $\psi(x)=Q_{2m+1}(x)$, we can get
\begin{align*}
\mathcal{A}_{\rm l}Q_{2m+1}(y)=u_m\left(\frac{p_{2m}(y)}{h_{2m}}+\frac{t_{2m+2,2m}}{h_{2m}}p_{2m+2}(y)\right).
\end{align*}

To summarise, we have the following proposition.
\begin{proposition}
In the setting of Proposition \ref{P3.9},
one has the  following relations between the discrete orthogonal polynomials $\{p_n(x)\}_{n=0}^\infty$ and discrete skew orthogonal polynomials $\{Q_n(x)\}_{n=0}^\infty$
\begin{subequations}
\begin{align}
&(1)\,\,p_{2m+1}(x)=Q_{2m+1}(x);\label{pqodd}\\
&(2)\,\,\frac{1}{u_m}\mathcal{A}_{\rm l}Q_{2m+1}(y)=\frac{p_{2m}(y)}{h_{2m}}+t_{2m+2,2m}\frac{p_{2m+2}(y)}{h_{2m+2}};\label{dqodd}\\
&(3)\,\,p_{2m}(x)=Q_{2m}(x)+t_{2m,2m-2}Q_{2m-2}(x);\label{pqeven}\\
&(4)\,\,\frac{1}{u_m}\mathcal{A}_{\rm l}Q_{2m}(y)=-\frac{p_{2m+1}(y)}{h_{2m+1}}\label{dqeven}.
\end{align}
\end{subequations}
\end{proposition}

To make use of these in relation to the correlation kernels, consider first the discrete unitary Christoffel-Darboux kernel \eqref{cd} with $N \mapsto 2N -1$, and 
rewrite to read
\begin{align*}
K_{2N-1}(x,y)=\sum_{m=0}^{N-1}\frac{1}{h_{2m}}p_{2m}(x)p_{2m}(y)+\sum_{m=0}^{N-1}\frac{1}{h_{2m+1}}p_{2m+1}(x)p_{2m+1}(y).
\end{align*} 
Use of \eqref{pqodd} and \eqref{dqeven} shows
\begin{align*}
\sum_{m=0}^{N-1}\frac{1}{h_{2m+1}}p_{2m+1}(x)p_{2m+1}(y)=-\sum_{m=0}^{N-1}\frac{1}{u_m}Q_{2m+1}(x)\mathcal{A}_{\rm l}Q_{2m}(y).
\end{align*}
Similarly, by using both \eqref{dqodd} and \eqref{pqeven}, one can obtain
\begin{align*}
\sum_{m=0}^{N-1}\frac{1}{h_{2m}}p_{2m}(x)p_{2m}(y)=\sum_{m=0}^{N-1}\frac{1}{u_m}Q_{2m}(x)\mathcal{A}_{\rm l}Q_{2m+1}(y)-\frac{t_{2N,2N-2}}{h_{2N}}p_{2N}(y)Q_{2N-2}(x).
\end{align*}
Combining these gives 
\begin{align}
K_{2N-1}(x,y)=\mathcal{A}^{(y)}_{\rm l}S_N(x,y)-\frac{t_{2N,2N-2}}{h_{2N}}p_{2N}(y)Q_{2N-2}(x),
\end{align}
 where $S_N(x,y)$ is the Christoffel-Darboux kernel of the discrete symplectic ensemble defined in \eqref{scd}, and $\mathcal{A}^{(y)}_{\rm l}$ means the operator acts on the variable $y$.

Use of the discrete Pearson equation \eqref{dp} allows the operator $\mathcal{A}_{\rm l}$ to be rewritten
\begin{align*}
\mathcal{A}_{\rm l}&:=g(x)T+f(x)(\Delta+\nabla)=\rho^{-1}(x)[\rho(x+1)f(x+1)T-\rho(x)f(x)T^{-1}]\\&:=\rho^{-1}(x)(\omega(x+1)T-\omega(x)T^{-1}):=\rho^{-1}(x)\mathcal{R},
\end{align*}
where $T^{-1}$ is defined as $T^{-1}\phi(x)=\phi(x-1)$. 
Furthermore, as suggested by \cite[Equation (2.3)]{borodin09}, by inspection there exists an inverse  for the operator
$\mathcal{R}:= (\omega(x+1)T-\omega(x)T^{-1})$, called $\epsilon$, whose action is different with the even lattices and odd lattices as specified by
 \begin{align*}
( \epsilon\cdot\phi)(2m)&=-\sum_{k=m}^{\infty}\frac{\omega(2m+2)\cdots\omega(2k)}{\omega(2m+1)\cdots\omega(2k+1)}\phi(2k+1),\\ (\epsilon\cdot\phi)(2m+1)&=\sum_{k=-\infty}^{m}\frac{\omega(2k+2)\cdots\omega(2m)}{\omega(2k+1)\cdots\omega(2m+1)}\phi(2k).
 \end{align*}
Therefore, the inverse of $\mathcal{A}_{\rm l}$ can be written as $\epsilon\rho$, and we denote it as $\mathcal{D}$.

 With this notation, firstly, we can rewrite the equation \eqref{dqeven} as
 \begin{align*}
 Q_{2m}(y)=-\frac{u_m}{h_{2m+1}}(\mathcal{D}\cdot p_{2n+1})(y),
 \end{align*}
 and moreover, we have the following proposition, which coincides with the result in \cite[Corollary 2.8]{borodin09}.
 \begin{proposition}\label{P3.12}
 In the setting of Proposition \ref{P3.9} we have the
relation between the Christoffel-Darboux kernels of the discrete unitary ensemble and discrete symplectic ensemble
\begin{align*}
 S_N(x,y)= \mathcal{D}^{(y)}  K_{2N-1}(x,y)-\frac{c_{2N-1}u_{N-1}}{c_{2N-2}h_{2N}h_{2N-1}}(\mathcal{D}\cdot p_{2N})(y)(\mathcal{D}\cdot p_{2N-1})(x).
 \end{align*}
 \end{proposition}

 Therefore, by taking the explicit $c_N$ in Section \ref{example1}, the kernels of Meixner, Charlier and Hahn polynomials can be obtained as
 \begin{align*}
& S_N^{(Meixner)}(x,y)=\mathcal{D}^{(y)}K^{(Meixner)}_{2N}-\frac{1-a}{h^{(Meixner)}_{2N-1}}(\mathcal{D}\cdot p_{2N})(y)(\mathcal{D}\cdot p_{2N-1})(x),\\
& S_N^{(Charlier)}(x,y)=\mathcal{D}^{(y)}K^{(Charlier)}_{2N}-\frac{1}{h^{(Charlier)}_{2N-1}}(\mathcal{D}\cdot p_{2N})(y)(\mathcal{D}\cdot p_{2N-1})(x),\\
 & S_N^{(Hahn)}(x,y)=\mathcal{D}^{(y)}K^{(Hahn)}_{2N}-\frac{2N+\al+\be+1}{h^{(Hahn)}_{2N-1}}(\mathcal{D}\cdot p_{2N})(y)(\mathcal{D}\cdot p_{2N-1})(x).
 \end{align*}

 \section{The $q$-analogy to symplectic ensemble}
 The choice $w ( x) = x^a (1 - x)^b$, supported on $0 < x < 1$ in (\ref{1.4}) defines the Selberg weight in the theory of the Selberg integral;
 see e.g.~\cite[Section 4]{forrester10} and \cite{FW08}. As noted in the latter references, there are natural generalisation of the continuous Selberg weight and integral to
 a discretisation on the exponential lattice known as the  $q$-Selberg weight and integral. We will use this in the case corresponding to $\beta=4$ to motivate
 a study of discrete symplectic ensembles on the exponential lattice.
 
%

Throughout this section, we fix $q$ as $0<q<1$ and use the symbols 
\begin{align*}
(a;q)_N=\prod_{i=0}^{N-1}(1-aq^i),\quad (a;q)_\infty=\prod_{i=0}^\infty(1-aq^i)
\end{align*}
and
\begin{align*}
(a;q)_\alpha = { (a;q)_\infty \over (a q^\alpha; q)_\infty  }.
\end{align*}
We will require the $q$-analogous of gamma function defined as
\begin{align*}
\Gamma_q(x)=(1-q)^{1-x}\frac{(q;q)_\infty}{(q^x;q)_\infty},
\end{align*}
the $q$-difference operator 
\begin{align}\label{qdifference}
D_qf(x)=\frac{f(x)-f(qx)}{(1-q)x}
\end{align}
and the particular $q$-Jackson integral (cf.~(\ref{J}))
\begin{align}\label{J1}
\int_0^a f(z)d_qz:=(1-q)\sum_{n=0}^\infty aq^nf(aq^n).
\end{align}
Also required is the multiple integral generalisation of the latter
\begin{align*}
\int_0^1 \cdots \int_{z_n=0}^{q^{\gamma}z_{n-1}} f(z)\frac{d_q z_1}{z_1}\cdots\frac{d_q z_n}{z_n}:=(1-q)^n\sum_{\langle\xi_F\rangle}f(t_1,\cdots,t_n)
\end{align*}
with the summation region $\langle\xi_F\rangle$ 
\begin{align*}
t_1=q^{\mu_1},\, t_2/t_1=q^{\mu_2}q^{\gamma}, \,\cdots, t_n/t_{n-1}=q^{\mu_n}q^{\gamma},\quad \mu_j\in\mathbb{Z}_{\geq0}.
\end{align*}

\subsection{The model of $q$-symplectic ensemble}
We begin with the $q$-generalisation of the Selberg integral, given by Aomoto   \cite{aomoto98}, 
\begin{align*}
\begin{aligned}
\int_{z_1=0}^1\cdots\int_{z_n=0}^{q^\gamma z_{n-1}}&\prod_{i=1}^nz_i^\alpha\frac{(qz_i;q)_\infty}{(q^\beta z_i;q)_\infty}\prod_{1\leq j<k\leq n}z_j^{2\gamma-1}\frac{(q^{1-\gamma}z_k/z_j;q)_\infty}{(q^\gamma z_k/z_j;q)_\infty}(z_j-z_k)\frac{d_qz_n}{z_n}\cdots\frac{d_qz_1}{z_1}\\
&=q^{\al\gamma{n\choose 2}+2\gamma^2{n\choose 3}}\prod_{j=1}^n\frac{\Gamma_q(\alpha+(j-1)\gamma)\Gamma_q(\be+(j-1)\gamma)\Gamma_q(j\gamma)}{\Gamma_q(\alpha+\be+(n+j-2)\gamma)\Gamma_q(\gamma)}
\end{aligned}
\end{align*}
for $\al,\be,\gamma\in\mathbb{C}$ satisfying $|q^{\al+(i-1)\gamma}|<1$ for $i=1,\cdots,n$.

This has the property that when $\gamma$ is a positive integer it reduces to the formula \cite{askey80,kadell88,ito17}
\begin{align}\label{Ao1}
\begin{aligned}
\frac{1}{n!}\int_0^1\cdots\int_0^1\prod_{i=1}^n z_i^{\alpha-1}(qz_i;q)_{\beta-1}\prod_{\substack{1\leq j<k\leq n\\1-\gamma\leq l\leq \gamma-1}}(z_j-q^l z_k)\prod_{1\leq j<k\leq n}(z_j-z_k)d_qz_1\cdots d_qz_n\\
=q^{\alpha\gamma {n\choose 2}+2\gamma^2 {n\choose 3}}\prod_{j=1}^n\frac{\Gamma_q(\alpha+(j-1)\gamma)\Gamma_q(\beta+(j-1)\gamma)\Gamma_q(j\gamma)}{\Gamma_q(\alpha+\beta+(n+j-2)\gamma)\Gamma_q(\gamma)}.
\end{aligned}
\end{align}
This latter form suggests a joint probability measure generalising (\ref{1.4}) in the case $\beta = 4$,
\begin{align}\label{q4}
\frac{1}{\tau_{2n}}\prod_{1\leq j<k\leq n}(x_j-q^{-1}x_k)(x_j-x_k)^2(x_j-qx_k)\prod_{i=1}^n\omega(x_i;q)d_qx_i, 
\end{align}
on the configuration space $I^n\subset\mathbb{R}^n$. Here
$\omega(x;q)$ is some proper $q$-weight function and $\tau_{2n}$ is the partition function given by
\begin{align}\label{q4a}
\tau_{2n}=\frac{1}{n!}\int_{I^n} \prod_{1\leq j<k\leq n}(x_j-q^{-1}x_k)(x_j-x_k)^2(x_j-qx_k)\prod_{i=1}^n\omega(x_i;q)d_qx_i, 
\end{align}
Use of the equality  
\begin{align}\label{q4b}
\prod_{1\leq j<k\leq n}(x_j-x_k)^2(x_j-q^{-1}x_k)(x_j-qx_k)
=(-q)^{-{n \choose 2}}\prod_{1\leq j<k\leq n}(x_j-x_k)^2(x_k-qx_j)(x_j-qx_k)
\end{align}
in (\ref{q4}) gives us a more symmetric form to work with.

We begin by expressing $\tau_{2n}$ as a Pfaffian.

\begin{lemma}
Let $\tau_{2n}$ be given by (\ref{q4a}). We have
\begin{align}\label{tridiagonal}
\tau_{2n}=\Pf\left[\int_I (x^iD_qx^j-x^jD_q x^i)\omega(x;q)d_qx\right]_{i,j=0}^{2n-1}.
\end{align}
\end{lemma}

\begin{proof}
We first rewrite the product of differences in  (\ref{q4b}).
Using the Vandermonde determinant
\begin{align*}
\prod_{1\leq j<k\leq 2n}(y_k-y_j)=\det[y_j^{k-1}]_{j,k=1}^{2n}
\end{align*}
and setting $y_j=x_j$ for $j=1,\cdots,n$ and $y_{n+j}=qx_j$ for $j=1,\cdots,n$, one can see that the left hand side of the Vandermonde determinant is equal to
\begin{align*}
\prod_{1\leq j<k\leq n}&(x_k-x_j)q(x_k-x_j)\prod_{j,k=1}^n(qx_k-x_j)\\
&=q^{n \choose 2}\prod_{1\leq j<k\leq n}(x_k-x_j)^2(qx_k-x_j)(qx_j-x_k)\prod_{j=1}^n (qx_j-x_j).
\end{align*}
By using this formula, we see 
\begin{align*}
\prod_{1\leq j<k\leq n}(x_j-x_k)^2(x_j-qx_k)(x_k-qx_j)=(-q)^{-{n \choose 2}}\det\left[
x_i^j,\, D_qx_i^j
\right]_{i=1,\cdots,n}^{j=0,\cdots,2n-1},
\end{align*}
where $D_q$ is the $q$-difference operator defined by \eqref{qdifference}. The stated result now follows from the de Bruijn formula \cite{debruijn55,forrester18}.
\end{proof}

The most systematic way to evaluate the Pfaffian in \eqref{tridiagonal} is to tridiagonalise the skew symmetric moment matrix, which leads us to consider  \textit{$q$-skew orthogonal polynomials}. In next subsection, we will give the definition of the $q$-skew orthogonal polynomials and formulate the correlation kernel in terms of the $q$-skew orthogonal polynomials, using a method similar to that given in the Section \ref{dsop1} and Section \ref{dsop2}.

\subsection{$q$-skew orthogonal polynomials and correlation function}
We denote $\mathbb{R}_q[x]$ as the ring of polynomials in $x$ over the field $\mathbb{R}(q)$. Consider the skew symmetric inner product $\langle\cdot,\cdot\rangle_{s,\omega}$: $\mathbb{R}_q[x]\times\mathbb{R}_q[x]\mapsto\mathbb{R}(q)$, with the form
\begin{align}\label{sinp}
\langle f(x;q),g(y;q)\rangle_{s,\omega}= \int_I (f(x;q)D_qg(x;q)-g(x;q)D_qf(x;q))\omega(x;q)d_qx.
\end{align}
In terms of this inner product one has for the sequence of moments
\begin{align}\label{qm}
m_{i,j}(q)=\langle x^i,y^j\rangle_{s,\omega}=([j]_q-[i]_q)\int_I x^{i+j-1}\omega(x;q)d_qx,\quad [i]_q=\frac{1-q^i}{1-q}.
\end{align}
Assuming the even-order minor principles of moment matrix $M_\infty(q)=(m_{i,j}(q))_{i,j=0}^\infty$ are nonsingular (i.e. $\tau_{2n}\not=0$ for $n\in\mathbb{Z}_{\geq0}$), we can apply the skew-Borel decomposition to this moment matrix
to obtain
\begin{align}\label{qsbd}
M_\infty(q)=S^{-1}(q)J(q)S^{-\top}(q).
\end{align}
Here $S(q)$ is a  lower triangular matrix with diagonals $1$ and $J(q)$ a skew symmetric $2\times 2$ block matrix, denoted as
\begin{align*}
J(q)=\left(\begin{array}{ccccc}
0&u_0(q)&&&\\
-u_0(q)&0&&&\\
&&0&u_1(q)&\\
&&-u_1(q)&0&\\
&&&&\ddots
\end{array}\right).
\end{align*} 
By introducing $\chi(x)=(1,x,x^2,\cdots)^\top$ we can therefore define the $q$-skew orthogonal polynomials $Q(x;q):=\{Q_n(x;q)\}_{n=0}^\infty$ by
\begin{align*}
Q(x;q)=S(q)\chi(x),
\end{align*}
satisfying the defining properties
\begin{align}\label{qsop}
\begin{aligned}
&\langle Q_{2n}(x;q),Q_{2m+1}(x;q)\rangle_{s,\omega}=-\langle Q_{2m+1}(x;q),Q_{2n}(x;q)\rangle_{s,\omega}=u_n(q)\delta_{n,m},\\
&\langle Q_{2m}(x;q), Q_{2n}(x;q)\rangle_{s,\omega}=\langle Q_{2m+1}(x;q),Q_{2n+1}(x;q)\rangle_{s,\omega}=0.
\end{aligned}
\end{align}
In this setting one can compute $u_n(q)$ in \eqref{qsop} as $u_n(q)=\frac{\tau_{2n+2}}{\tau_{2n}}$. Moreover, $\tau_{2n}$ defined in \eqref{tridiagonal} can then be written as
\begin{align*}
\Pf\left[\int_I(x^iD_qx^j-x^jD_qx^i)\omega(x;q)d_qx\right]_{i,j=0}^{2n-1}=\Pf\left[\langle Q^s_i(x;q),Q^s_j(y;q)\rangle_{s,\omega}\right]_{i,j=0}^{2n-1}=\prod_{i=0}^{n-1}u_i=\tau_{2n}.\\
\end{align*}

It is also the case that this family of polynomials have the following $q$-integral representation (cf.~\cite{Ey01}, \cite{Go09}).
\begin{proposition}
In terms of multidimensional $q$-integrals,
\begin{align*}
Q_{2n}(x;q)&=\frac{1}{n!\tau_{2n}}\int_{I^n}\Delta_q^4(x)\prod_{i=1}^n (x-x_i)(x-qx_i)\omega(x_i;q)\dq_i,\\
Q_{2n+1}(x;q)&=\frac{1}{n!\tau_{2n}}\int_{I^n}\Delta_q^4(x)\left(x+(1+q)\sum_{i=1}^nx_i+c\right)\prod_{i=1}^n(x-x_i)(x-qx_i)\omega(x_i;q)\dq_i
\end{align*}
with $\Delta_q^4(x):=\prod_{1\leq j<k\leq n}(x_j-x_k)^2(x_j-qx_k)(q^{-1}x_k-x_j)$ and $\tau_{2n}$ being defined in \eqref{tridiagonal}. Moreover, $c$ in the second equality is an arbitrary constant as is $\gamma_{2m}$ in \eqref{1.10}.
\end{proposition}

Now, we turn to the correlation function of the $q$-symplectic case. Define the correlation function on the phase space $I^k$, where $I$ is the support of $\omega(x;q)$, as
\begin{align*}
\rho_{n,k}(x_1,\cdots,x_k)=\frac{1}{(n-k)!\tau_{2n}}\int_{I^{n-k}}\Delta_q^4(x)\prod_{i=k+1}^{n}\omega(x_i;q)d_qx_i.
\end{align*}
According to \cite[Corollary 1.4]{rains00}, if we take the measure as the discrete measure on the exponential lattice, then this correlation function can be written as a Pfaffian 
specified by a particular $2\times 2$ skew symmetric kernel.

\begin{proposition}
The statistical state corresponding to the probability measure (\ref{q4}) is a Pfaffian point process.
Explicity, the $k$-point correlation function $\rho_{n,k}$ admits the form 
$$\rho_{n,k}= \prod_{l=1}^k \omega(x_l;q) \Pf(\tilde{R}_n(x_i,x_j))_{i,j=1}^k,$$ where
\begin{align*}
 \tilde{R}_n(x,y)=\left(
\begin{array}{cc}
\sum_{0\leq i,j\leq 2n-1}x^iM(q)^{-\top}_{i,j}y^j&\sum_{0\leq i,j\leq {2n-1}}x^iM(q)^{-\top}_{i,j}D_qy^j\\
\sum_{0\leq i,j\leq 2n-1}D_qx^iM(q)^{-\top}_{i,j}y^j&\sum_{0\leq i,j\leq 2n-1}D_qx^iM(q)^{-\top}_{i,j}D_qy^j
\end{array}
\right),
\end{align*}
with $M(q)$ the moment matrix with entries $m_{i,j}(q)$ in \eqref{qm}.
\end{proposition}
 Analogous to (\ref{ss}), if we denote 
 \begin{align*}
R_n(x,y;q):=\sum_{0\leq i,j\leq 2n-1}x^iM(q)^{-\top}_{i,j}y^j=\sum_{i=0}^{n-1}\frac{1}{u_i(q)}\left(Q_{2i}(x;q)Q_{2i+1}(y;q)-Q_{2i}(y;q)Q_{2i+1}(x;q)\right),
 \end{align*}
we can write
\begin{align*}
\tilde{R}_n(x,y)=\left(
\begin{array}{cc}
R_n(x,y;q)&D_{q,y} R_n(x,y;q)\\
D_{q,x}R_n(x,y;q)&D_{q,x}D_{q,y}R_n(x,y;q)
\end{array}
\right).
\end{align*}
 
\subsection{Two examples of classical $q$-skew orthogonal polynomials} \label{exl}
To guide our study in the more general case, we begin by analysing two specific examples of 
$q$-skew orthogonal polynomials as examples. A general discussion will then be given in the next subsection.

\subsubsection{Al-Salam $\&$ Carlitz type skew-orthogonal polynomials}
Consider the Al-Salam $\&$ Carlitz weight \cite{baker00,koekoek96}
\begin{align}\label{wAC}
\omega(x;q)=\frac{(qx;q)_\infty(qx/\alpha;q)_\infty}{(q;q)_\infty(\alpha;q)_\infty(q/\alpha;q)_\infty},\quad  \alpha<0
\end{align}
and its corresponding integral region $I=[\alpha,1]$, interpreted as
\begin{align}\label{J1}
\int_\alpha^1 f(z)d_qz:=(1-q) \Big ( \sum_{n=0}^\infty q^nf(q^n) - \alpha  \sum_{n=0}^\infty \alpha q^nf(\alpha q^n) \Big ).
\end{align}
 We say $\{U_n^{(\alpha)}(x;q)\}_{n=0}^\infty$ are the monic Al Salam $\&$ Carlitz orthogonal polynomials with weight function $\omega(x;q)$ in the integral $I$ if they satisfy 
\begin{align}\label{qorthogonal}
\int_\alpha^1U_m^\al(x;q)U_n^\al(x;q)\omega(x;q)\dq=(1-q)(-\alpha)^nq^{n\choose 2}(q;q)_n\delta_{n,m}:=h_n\delta_{n,m}.
\end{align}

It is known that these polynomials satisfy the following lowering equation (see e.g.~\cite{AS93})
\begin{align}\label{qhermite}
D_q U_n^\al(x;q)=[n]_qU_{n-1}^\al(x;q),\quad [n]_q=\frac{1-q^n}{1-q}.
\end{align}
Being monic polynomials of successive degree, they form basis for the polynomials space $\mathbb{R}_q[x]$, and thus one can expand the 
corresponding  skew orthogonal polynomials $\{Q_{m}^\al(x;q)\}_{m=0}^\infty$ in the form
\begin{align*}
Q_{2m}^\al(x;q)=\sum_{p=0}^{2m} a_{m,p}U_{2m-p}^\al(x;q),\quad Q_{2m+1}^\al(x;q)=\sum_{p=0}^{2m+1}b_{m,p}U_{2m+1-p}^\al(x;q)
\end{align*} 
with $a_{m,0}=b_{m,0}=1$. We seek the explicit form of the coefficients. 

Firstly, consider polynomials of even degree. From the skew orthogonal relation \eqref{qsop}, we know
\begin{align*}
\langle Q_{2m}^\al(x;q), U_i^\al(x;q)\rangle_{s,\omega}=0,\quad 0\leq i\leq 2m,
\end{align*}
or equivalently,
\begin{align*}
\int_I(Q_{2m}^\al(x;q)D_q U_i^\al(x;q)-D_q Q_{2m}^\al(x;q)U_i^\al(x;q))\omega(x;q)\dq=0.
\end{align*}
Then from the lowering relation \eqref{qhermite} and orthogonality \eqref{qorthogonal} of the Al-Salam $\&$ Carlitz orthogonal polynomials, one can obtain
\begin{align*}
a_{m,2m-i+1}[i]_qh_{i-1}-a_{m,2m-i-1}[i+1]_qh_i=0,\quad 0\leq i\leq 2m.
\end{align*}
This  is an explicit iterative relation for the coefficients $\{a_{m,p},\ p=0,\cdots,2m\}$.
Inserting the explicit value of $h_i$, it reads
\begin{align*}
a_{m,2m-i+1}=-\al(1-q)q^{i-1}[i+1]_qa_{m,2m-i-1},\quad 0\leq i\leq 2m.
\end{align*}
Notice that $a_{m,-1}=0$ implies $a_{m,\text{odd}}=0$, while $a_{m,0}=1$ implies
\begin{align*}
a_{m,2p}=\prod_{l=1}^p\{-\al(1-q)q^{2m-2l}[2m-2l+2]_q\}=\prod_{l=1}^p \{-\al(1-q^2)q^{2(m-l)}[m-l+1]_{q^2}\}.
\end{align*}
Changing the variable $2p\mapsto 2(m-p)$, one can show 
\begin{align*}
a_{m,2(m-p)}=\prod_{l=1}^{m-p}\{-\al(1-q^2)[m-l+1]_{q^2}q^{2(m-l)}\}=(-\al(1-q^2))^{m-p}q^{(m-p)(m+p-1)}\frac{[m]_{q^2}!}{[p]_{q^2}!},
\end{align*}
and thus 
\begin{align*}
Q_{2m}^\al(x;q)=(-\al(1-q^2))^m[m]_{q^2}!\sum_{p=0}^m\frac{q^{(m-p)(m+p-1)}}{(-\al(1-q^2))^p[p]_{q^2}!}U_{2p}^\al(x;q).
\end{align*}

For polynomials of odd degree, we make particular use of the skew orthogonal relations
\begin{subequations}
\begin{align}
&\langle Q_{2m+1}^\al(x;q), Q_{2m+1}^\al(x;q)\rangle_{s,\omega}=0,\label{qsk1}\\
&\langle Q_{2m+1}^\al(x;q), U_i^\al(x;q)\rangle_{s,\omega}=0,\quad 0\leq i\leq 2m-1.\label{qsk2}
\end{align}
\end{subequations}
The equation \eqref{qsk2} gives rise to
\begin{align}\label{recurrence}
b_{m,2m-i+2}[i]_qh_{i-1}-b_{m,2m-i}[i+1]_qh_i=0,\quad 0\leq i\leq 2m-1
\end{align}
by using the orthogonal relation \eqref{qorthogonal}. Also, noting that $b_{m,2m+2}=0$, it follows
\begin{align}\label{b}
b_{m,2m+2}=\cdots=b_{m,2}=0.
\end{align}
The equation \eqref{qsk1} is equivalent to
\begin{align*}
\langle  Q_{2m+1}^\al(x;q),U_{2m+1}^\al(x;q)+b_{m,1}U_{2m}^\al(x;q)\rangle_{s,\omega}=0.
\end{align*}
By using \eqref{qsk2} agian, it follows that
\begin{align*}
b_{m,1}[2m+1]_qh_{2m}+b_{m,2}b_{m,1}[2m]_qh_{2m-1}-b_{m,1}[2m+1]_qh_{2m}=0.
\end{align*}
This shows $b_{m,1}=0$ or $b_{m,2}=0$ or both of them are equal to zero. Without generality, we can assume $b_{m,1}$ is not equal to zero ($b_{m,2}=0$ has been obtained in \eqref{b}). Further, one can see the iterative process \eqref{recurrence} is the same as for $a_{m,2p}$. Therefore we get
\begin{align*}
Q_{2m+1}^\al(x;q)=U_{2m+1}^\al(x;q)+b_{m,1}Q_{2m}^\al(x;q).
\end{align*}
Using the freedom implied by (\ref{1.10}) we can take $b_{m,1}=0$ for brevity.

In summary, we have found that 
\begin{align}\label{alsalam}
\begin{aligned}
&Q_{2m+1}^\al(x;q)=U_{2m+1}^\al(x;q),\\ &Q_{2m}^\al(x;q)=(-\al(1-q^2))^m[m]_{q^2}!\sum_{p=0}^m\frac{q^{(m-p)(m+p-1)}}{(-\al(1-q^2))^p[p]_{q^2}!}U_{2p}^\al(x;q),
\end{aligned}
\end{align}
are $q$-skew orthogonal with respect to the skew symmetric inner product \eqref{sinp}
under the Al-Salam $\&$ Carlitz weight (\ref{wAC}).  Moreover, the normalisation constant can be evaluated by
\begin{align*}
&\frac{1}{n!}\int_I\prod_{1\leq j<k\leq n}(x_j-q^{-1}x_k)(x_j-x_k)^2(x_j-qx_k)\prod_{i=1}^n\omega(x_i;q)d_qx_i\\
&\quad =\prod_{i=0}^{n-1}\langle Q_{2i}^\al(x;q),Q_{2i+1}^\al(x;q)\rangle=\prod_{i=0}^{n-1}\{c_{2i}[2i+1]_q\}=\frac{1}{2^n}\alpha^{n(n-1)}q^{\frac{n(2n-1)(n-1)}{6}}\prod_{i=0}^{n-1}(q;q)_{2i+1},
\end{align*}
which corresponds to the results in \cite[Equation 4.27]{baker00}.

\subsubsection{Little $q$-Jacobi skew orthogonal polynomials}\label{lqj}
The monic little $q$-Jacobi polynomials $\{p_n^\ab(x;q)\}_{n=0}^\infty$ are defined by the orthogonality relation \cite{masuda91,koekoek96}
\begin{align*}
\int_0^1 p_n^\ab(x;q)x^k\omega^\ab(x;q)d_qx=0,\quad 0\leq k\leq n-1
\end{align*}
with the weight function
\begin{align}\label{weight}
\omega^\ab(x;q):=(qx;q)_\beta x^\alpha,\quad \alpha>-1,\,\beta>-1.
\end{align}
This is the same weight as in (\ref{Ao1}) upon incrementing $\alpha$ and $\beta$ by 1 in the latter, and one can easily find
\begin{align*}
\omega^\ab(x;q)=x(1-q^\be x)\omega^{(\al-1,\be-1)}(x;q).
\end{align*}
For the normalisation, one has
\begin{align}\label{qj}
\begin{aligned}
&\int_0^1 p_n^\ab(x;q)p_m^\ab(x;q)\omega^\ab(x;q)\dq\\
&=q^{n(n+\alpha+2)}[\alpha+\beta+2n+1]_q^{-1}\frac{(q;q)_{n+\al+\be}(q;q)_{n+\al}(q;q)_{n+\be}(q;q)_n}{(q;q)_{2n+\al+\be}^2}\delta_{n,m}=:h_n^{(\al,\be)}\delta_{n,m}.
\end{aligned}
\end{align}

It is well known that the little $q$-Jacobi polynomials permit the lowering operation \cite{masuda91}
\begin{align}\label{qjacobi}
D_q p_n^\ab(x;q)=[n]_qp_{n-1}^{(\alpha+1,\beta+1)}(x;q)
\end{align}
and the raising operation
\begin{align*}
D_p\left[\omega^\ab(x;q)p_n^\ab(x;q)\right]=-\frac{1-q^{n+\alpha+\be}}{(1-q)q^{n+\al-1}}\omega^{(\al-1,\be-1)}(x;q)p_{n+1}^{(\al-1,\be-1)}(x;q),\quad p=\frac{1}{q}.
\end{align*}
From this the Rodrigues formula enables us to give an explicit expression for this polynomial,
\begin{align*}
p_n^\ab(x;q)=\frac{q^{(n+\al)n}(q^{-n-\al};q)_n}{(q^{n+\al+\be+1};q)_n}\sum_{k=0}^n\frac{(q^{-n};q)_k(q^{n+\al+\be+1};q)_k}{(q^{\al+1};q)_k}\frac{q^kx^k}{(q;q)_k}=:x^n+\sum_{i=1}^n\gamma_{n,i}^\ab x^{n-i}.
\end{align*}
Note in particular
\begin{align*}
\gamma_{n,1}^\ab=-\frac{(1-q^n)(1-q^{n+\al})}{(1-q)(1-q^{2n+\al+\be})},
\end{align*}
which will be used later.

To construct the relations between little $q$-Jacobi orthogonal polynomials and $q$-skew orthogonal polynomials, the following proposition is needed.
\begin{lemma}
One has the relation between $\{p_n^{(\al-1,\be-1)}(x;q)\}_{n=0}^\infty$ and $\{p_n^\ab(x;q)\}_{n=0}^\infty$, 
\begin{align*}
p_n^{(\al-1,\be-1)}(x;q)=p_n^\ab(x;q)+a_{n,n-1}^\ab p_{n-1}^\ab(x;q)+a_{n,n-2}^\ab p_{n-2}^\ab(x;q),
\end{align*}
where $a_{n,n-1}^\ab$ and $a_{n,n-2}^\ab$ are given in \eqref{an}.
\end{lemma}
\begin{proof}
Since $\{p_n^\ab(x;q)\}_{n=0}^\infty$ form a basis of $\mathbb{R}_q[x]$, one can assume
\begin{align*}
p_n^{(\al-1,\be-1)}(x;q)=\sum_{k=0}^n a_{n,k}^\ab p_k^\ab(x;q).
\end{align*}
From the orthogonal relation \eqref{qj}, one knows
\begin{align*}
&\int_0^1 p_n^{(\al-1,\be-1)}(x;q) x^i\omega^{(\al-1,\be-1)}(x;q)d_qx\\
&=\sum_{k=0}^n a_{n,k}^\ab\int_0^1 p_k^\ab(x;q)x^i\omega^{(\al-1,\be-1)}(x;q)\dq=0,\quad 0\leq i\leq n-1.
\end{align*}
Moreover, the relation \eqref{weight} leads to $a_{n,0}=\cdots=a_{n,n-3}=0$, and hence $$p_n^{(\al-1,\be-1)}(x;q)=p_n^\ab(x;q)+a_{n,n-1}^\ab p_{n-1}^\ab(x;q)+a_{n,n-2}^\ab p_{n-2}^\ab(x;q).$$
Now, from the equations
\begin{align*}
&\int_0^1 p_n^{(\al-1,\be-1)}(x;q)x^n\omega^{(\al-1,\be-1)}(x;q)\dq=h_n^{(\al-1,\be-1)},\\
&\int_0^1 p_n^{(\al-1,\be-1)}(x;q)x^{n+1}\omega^{(\al-1,\be-1)}(x;q)\dq=-\gamma_{n+1,1}^{(\al-1,\be-1)}h_n^{(\al-1,\be-1)},
\end{align*}
it follows
\begin{align}
\begin{aligned}\label{ac}
&a_{n,n-2}^\ab h_{n-2}^\ab=-q^\beta h_n^{(\al-1,\be-1)},\\
&a_{n,n-1}^\ab h_{n-1}^{(\al,\be)}=c_{n}^{(\al-1,\be-1)}+q^\be h_n^{(\al-1,\be-1)}(\gamma_{n+1,1}^{(\al-1,\be-1)}-\gamma_{n-1,1}^\ab),
\end{aligned}
\end{align}
allowing us to compute
\begin{align}\label{an}
\begin{aligned}
&a_{n,n-2}^\ab=-q^{2n+2\al+\be-1}\frac{(1-q^{n+\al-1})(1-q^{n+\be-1})(1-q^{n-1})(1-q^n)}{(1-q^{\al+\be+2n-1})(1-q^{\al+\be+2n-2})^2(1-q^{\al+\be+2n-3})},\\
&a_{n,n-1}^\ab=\frac{q^{n+\al+1}(1-q^n)}{1-q^{n+\al+\be-1}}\left[1+\frac{q^\be}{1-q}\left(\frac{(1-q^{n-1})(1-q^{n+\al-1})}{1-q^{2n+\al+\be-2}}-\frac{(1-q^{n+1})(1-q^{n+\al})}{1-q^{2n+\al+\be}}\right)\right],
\end{aligned}
\end{align}
and thus complete the proof.
\end{proof}

Now consider the monic little $q$-Jacobi skew orthogonal polynomials $\{Q_m^\ab(x;q)\}_{m=0}^\infty$, which have the expansion with the basis $\{p_m^{\ab}(x;q)\}_{m=0}^\infty$
\begin{align*}
&Q_{2m}^\ab(x;q)=\sum_{j=0}^{2m}\xi^\ab_{m,j}p_{j}^{\ab}(x;q),\quad Q_{2m+1}^\ab(x;q)=\sum_{j=0}^{2m+1}\eta^\ab_{m,j}p_{j}^{\ab}(x;q).
\end{align*}
Using the $q$-skew orthogonality relation
\begin{align*}
\int_0^1 \left(Q_{2m}^\ab(x;q)D_q p_i^\ab(x;q)-D_q Q_{2m}^\ab(x;q) p_i^\ab(x;q)\right)\omega^{(\al+1,\be+1)}(x;q)\dq=0,\quad 0\leq i\leq 2m
\end{align*}
and taking the lowering operation for the little $q$-Jacobi polynomials \eqref{qjacobi}, one can find
\begin{align*}
&\int_0^1 Q_{2m}^\ab(x;q)D_q p_i^\ab(x;q)\omega^{(\al+1,\be+1)}(x;q)\dq\\
&\quad =\sum_{j=0}^{2m}[i]_q\xi^\ab_{m,j}\int_0^1 p_j^{(\al,\be)}(x;q)p_{i-1}^{(\al+1,\be+1)}(x;q)\omega^{(\al+1,\be+1)}(x;q)\dq\\
&\quad =[i]_qh_{i-1}^{(\al+1,\be+1)}(\xi_{m,i-1}^\ab+a_{i,i-1}^{(\al+1,\be+1)}\xi_{m,i}^\ab+a_{i+1,i-1}^{(\al+1,\be+1)}\xi_{m,i+1}^\ab).
\end{align*}
A similar strategy shows
\begin{align*}
&\int_0^1 D_q Q_{2m}^\ab(x;q) p_i^\ab(x;q)\omega^{(\al+1,\be+1)}(x;q)\dq\\
&\quad =\sum_{j=0}^{2m} \xi^\ab_{m,j}[j]_q \int_0^1 p_{j-1}^{(\al+1,\be+1)}(x;q)p_i^\ab(x;q)\omega^{(\al+1,\be+1)}(x;q)\dq\\
&\quad =[i+1]_qh_i^{(\al+1,\be+1)}\xi_{m,i+1}^\ab+[i]_q h_{i-1}^{(\al+1,\be+1)}a_{i,i-1}^{(\al+1,\be+1)}\xi_{m,i}^\ab+[i-1]_q h_{i-2}^{(\al+1,\be+1)}a_{i,i-2}^{(\al+1,\be+1)}\xi_{m,i-1}^\ab.
\end{align*}
Moreover, the skew orthogonal relation implies
\begin{align*}
([i+1]_q &h_i^{{(\al+1,\be+1)}}-[i]_q h_{i-1}^{(\al+1,\be+1)} a_{i+1,i-1}^{(\al+1,\be+1)})\xi_{m,i+1}^\ab\\
&=([i]_q h_{i-1}^{(\al+1,\be+1)}-[i-1]_q h_{i-2}^{(\al+1,\be+1)} a_{i,i-2}^{(\al+1,\be+1)})\xi_{m,i-1}^\ab,\quad 0\leq i\leq 2m.
\end{align*}

Taking \eqref{ac} into the above equation and noting $\xi_{m,2m}^\ab=1$ gives
\begin{align*}
\xi_{m,2m-2j}^\ab=\prod_{k=1}^j \frac{[2m-2k+2]_qh_{2m-2k+1}^{{(\al+1,\be+1)}}+q^{\be+1}[2m-2k+1]_qh_{2m-2k+2}^\ab}{[2m-2k+1]_qh_{2m-2k}^{(\al+1,\be+1)}+q^{\be+1}[2m-2k]_qh_{2m-2k+1}^\ab}.
\end{align*}
The odd indexed coefficients vanish as follows by noting $\xi_{m,2m+1}^\ab=0$. Therefore,
\begin{align*}
Q_{2m}^\ab(x;q)=\sum_{j=0}^{m}\prod_{k=1}^{m-j}  \frac{[2m-2k+2]_qh_{2m-2k+1}^{{(\al+1,\be+1)}}+q^{\be+1}[2m-2k+1]_qh_{2m-2k+2}^\ab}{[2m-2k+1]_qh_{2m-2k}^{(\al+1,\be+1)}+q^{\be+1}[2m-2k]_qh_{2m-2k+1}^\ab}
p_{2j}^\ab(x;q).
\end{align*}

For the odd ones, firstly consider the skew orthogonal relation
\begin{align*}
\langle Q_{2m+1}^\al(x;q), p_i^\ab(x;q) \rangle_{s,\omega^{(\al+1,\be+1)}}=0,\quad 0\leq i\leq 2m-1.
\end{align*}
Following the above working, one can easily obtain the equation
\begin{align*}
([i+1]_q &h_i^{{(\al+1,\be+1)}}-[i]_q h_{i-1}^{(\al+1,\be+1)} a_{i+1,i-1}^{(\al+1,\be+1)})\eta_{m,i+1}^\ab\\
&=([i]_q h_{i-1}^{(\al+1,\be+1)}-[i-1]_q h_{i-2}^{(\al+1,\be+1)} a_{i,i-2}^{(\al+1,\be+1)})\eta_{m,i-1}^\ab,\quad 0\leq i\leq 2m-1.
\end{align*}
Since $\eta_{m,-1}^\ab=0$, this equation demonstrates that $\eta_{m,\rm odd}^\ab=0$ except $\eta_{m,2m+1}^\ab=1$. Moreover, the coefficients $\eta_{m,\rm even}^\ab$ enjoy the same recurrence relation with the coefficients of polynomials of even degree $\xi_{m,\rm even}^\ab$, which suggests the odd family can be chosen as 
\begin{align*}
Q_{2m+1}^\ab(x;q)=p_{2m+1}^\ab(x;q)+\gamma Q_{2m}^\ab(x;q)
\end{align*}
with arbitrary constant $\gamma$. For simplicity, we take $\gamma=0$.

Therefore, the skew orthogonal little $q$-Jacobi polynomials with weight function $\omega^{(\al+1,\be+1)}(x;q)$ in (\ref{sinp})
have the form
\begin{align*}
&Q_{2m+1}^\ab(x;q)=p_{2m+1}^\ab(x;q),\\
&Q_{2m}^\ab(x;q)=\sum_{j=0}^{m}\prod_{k=1}^{m-j}  \frac{[2m-2k+2]_qh_{2m-2k+1}^{{(\al+1,\be+1)}}+q^{\be+1}[2m-2k+1]_qh_{2m-2k+2}^\ab}{[2m-2k+1]_qh_{2m-2k}^{(\al+1,\be+1)}+q^{\be+1}[2m-2k]_qh_{2m-2k+1}^\ab}
p_{2j}^\ab(x;q).
\end{align*}

These two classical $q$-skew orthogonal polynomials reflect a structure similar to \eqref{disstr}, and so suggest a general theory relating classical $q$-skew orthogonal polynomials to classical $q$-orthogonal polynomials.
 
 \subsection{Relationship between $q$-skew orthogonal polynomials and $q$-orthogonal polynomials: General theorem}
The relationship of $q$-skew orthogonal polynomials and $q$-orthogonal polynomials is heavily dependent on the Pearson-type equation \eqref{qp}. Recall that in the classical case, the functions $f(x)$ and $g(x)$ in the Pearson-type equation \eqref{qp} are polynomials in $x$ of order $2$ and $1$ at most, respectively.

Define the operator 
\begin{align}\label{qa}
\mathcal{A}_q=q^{-\frac{1}{2}}g(x)T_q+q^{-1}f(x)D_{q^{-1}}+f(x)D_q,
\end{align} 
 where $T_q$ is the q-shifted operator defined by $T_q \phi(x)=\phi(qx)$, $D_q\phi(x)=\frac{\phi(x)-\phi(qx)}{(1-q)x}$ and $D_{q^{-1}}\phi(x)=\frac{\phi(x)-\phi(q^{-1}x)}{(1-q^{-1})x}$. The
 following analogue of Proposition \ref{P3.6} holds.
 
\begin{proposition}
Denote the inner product $\langle \cdot,\cdot\rangle$ as a symmetric bilinear form from $\mathbb{R}_q[x]\times \mathbb{R}_q[x]\mapsto\mathbb{R}(q)$, which is defined by
\begin{align*}
\langle \phi(x;q),\psi(x;q)\rangle=\int_I \phi(x;q)\psi(x;q)\rho(x;q)d_qx.
\end{align*}
One has
\begin{align}
&(1)\,\,\langle \mathcal{A}_q\phi(x;q),\psi(x;q)\rangle=-\langle \phi(x;q),\mathcal{A}_q\psi(x;q)\rangle;\label{qss}\\
&(2)\,\, \langle \phi(x;q), \mathcal{A}_q\psi(x;q)\rangle=\langle \phi(x;q),\psi(x;q)\rangle_{s, f(qx;q)\rho(qx;q)}\label{qos},
\end{align}
provided $f(x;q)\rho(x;q)$ vanishes at the end points of supports.
\end{proposition}
\begin{proof} 
 Here we just prove the case $I=[0,\infty)$; the other cases can be verified similarly.
 To do so, take the summation of these two inner products and find
\begin{align*}
&\langle \mathcal{A}_q\phi(x),\psi(x)\rangle+\langle \phi(x),\mathcal{A}_q\psi(x)\rangle
\\&=q^{-1}\int_0^\infty  f(x)\rho(x)(D_{q^{-1}}\phi(x)\psi(x)+\phi(x)D_{q^{-1}}\psi(x))d_qx+\int_0^\infty f(x)\rho(x) (D_q \phi(x)\psi(x)\\&+\phi(x)D_q\psi(x))d_qx+q^{-\frac{1}{2}}\int_0^\infty g(x)\rho(x)(\phi(qx)\psi(x)+\phi(x)\psi(qx))d_qx:=A_1+A_2+A_3.
\end{align*}
Then from the definition of $q$-Jackson integral \eqref{J1}, we know
\begin{align*}
A_3=(1-q)q^{-\frac{1}{2}}\sum_{n=-\infty}^\infty g(q^n)\rho(q^n)q^n\left(\phi(q^{n+1})\psi(q^n)+\phi(q^n)\psi(q^{n+1})\right) 
\end{align*}
and by noting the equation \eqref{qp} at the points $x=q^n$, this term can then be written as 
\begin{align*}
A_3=\sum_{n=-\infty}^{\infty}\left(\phi(q^{n+1})\psi(q^n)+\phi(q^n)\psi(q^{n+1})\right)\left(\rho(q^n)f(q^n)-\rho(q^{n+1})f(q^{n+1})\right).
\end{align*}
On the other hand, one can compute
\begin{align*}
A_1&=-2\sum_{n=-\infty}^{\infty} \phi(q^n)\psi(q^n)f(q^n)\rho(q^n)+\sum_{n=-\infty}^\infty f(q^n)\rho(q^n)\left(
\phi(q^{n-1})\psi(q^n)+\psi(q^{n-1})\phi(q^n)
\right),\\
A_2&=2\sum_{n=-\infty}^\infty \phi(q^n)\psi(q^n)f(q^n)\rho(q^n)-\sum_{n=-\infty}^\infty f(q^n)\rho(q^n)(\phi(q^{n+1})\psi(q^n)+\phi(q^n)\psi(q^{n+1})).
\end{align*}
If we combine these three equalities and assume $f(x;q)\rho(x;q)$ vanishes at the boundary points, then equation \eqref{qss} is proved.

For equation \eqref{qos}, it should be noted that
\begin{align*}
\langle \phi(x), \mathcal{A}_q\psi(x)\rangle&=q^{-\frac{1}{2}}\int_I g(x)\rho(x)\psi(qx)\phi(x)d_qx\\& +q^{-1}\int_I f(x)\rho(x)D_{q^{-1}}\psi(x)\phi(x)d_qx+\int_I f(x)\rho(x)\phi(x)D_q\psi(x)d_qx\\
&=-\int_I \frac{\psi(qx)\phi(x)f(qx)\rho(qx)}{(1-q)x}d_qx+\int_I\frac{\psi(q^{-1}x)\phi(x)f(x)\rho(x)}{(1-q)x}d_qx.
\end{align*}
Moreover, if $f(x;q)\rho(x;q)$ vanishes at the end points of supports, then 
\begin{align*}
\int_I\frac{\psi(q^{-1}x)\phi(x)f(x)\rho(x)}{(1-q)x}d_qx=\int_I \frac{\psi(x)\phi(qx)f(qx)\rho(qx)}{(1-q)x}d_qx,
\end{align*}
and therefore
\begin{align*}
\langle \phi(x),\mathcal{A}_q\psi(x)\rangle=\int_I \left(\phi(x)D_q\psi(x)-D_q\phi(x)\psi(x)\right)f(qx)\rho(qx)d_qx,
\end{align*}
as required.
 \end{proof}
 
 With the help of \eqref{qss} and orthogonal relation \eqref{qo}, one can see  (cf.~(\ref{1.12d}) and (\ref{1.12e}))
 \begin{align}\label{4.24a}
 \mathcal{A}_qp_n(x;q)=-\frac{c_n(q)}{\tih_{n+1}}p_{n+1}(x;q)+\frac{c_{n-1}(q)}{\tih_{n-1}}p_{n-1}(x;q),
 \end{align}
 where $c_n(q)$ is a constant dependent on Pearson pair $(f(x;q)$, $g(x;q))$ and the order of polynomials $n$ only, and $\tih_n$ is the normalisation constant defined in \eqref{qo} before. Similar to the derivation in \eqref{trans}-\eqref{disstr}, we finally get the sought relationship between $q$-orthogonal polynomials and $q$-skew orthogonal polynomials
 (cf.~Proposition \ref{P3.9}).
 
 \begin{proposition}\label{P4.6}
 Let $\rho(x;q)$ be a classical weight function in the sense of the Pearson-type equation \eqref{qp} with appropriate restriction on the degree of $f$ and $g$.
 Let the weight in (\ref{q4}) be related to $\rho(x)$ by $ \omega(x;q) = f(qx;q)\rho(qx;q)$. We have
\begin{subequations}
\begin{align}
&p_{2n+1}(x;q)=Q_{2n+1}(x;q),\quad p_{2n}(x;q)=Q_{2n}(x;q)-\frac{c_{2n-1}(q)}{c_{2n-2}(q)}Q_{2n-2}(x;q),\label{qpq}\\
&Q_{2n+1}(x;q)=p_{2n+1}(x;q),\quad Q_{2n}(x;q)=\left(\prod_{j=0}^{n-1}\frac{c_{2j+1}(q)}{c_{2j}(q)}\right)\sum_{l=0}^n\prod_{j=0}^{l-1}\frac{c_{2j}(q)}{c_{2j+1}(q)}p_{2l}(x;q)\label{qqp}
\end{align}
with normalisation $u_n(q) = c_{2n}(q)$.
\end{subequations}
\end{proposition}

\subsubsection{Several Examples}
In this part, we show the with explicit Pearson pairs $(f,g)$ and coefficients $c_n$ of Al-Salam $
\&$ Carlitz polynomials and Little $q$-Jacobi polynomials. The idea is to obtain the Pearson pairs from the weight functions with the Pearson-type equation \eqref{qp}, and from the equation \eqref{4.24a}, the latter quantities $c_n$ can then be given. 

\begin{enumerate}
\item Al-Salam $\&$ Carlitz polynomials. For the weight function \eqref{wAC}, one can obtain the Pearson pair 
\begin{align*}
(f,g)=\left(x^2-(1+\al)x+\al,q^{\frac{1}{2}}\frac{x-(1+\al)}{1-q}\right)
\end{align*}
from \eqref{qp}. The coefficients $c_n$ can then be obtained as
\begin{align*}
c_n=-\frac{q^{-n}}{1-q}h_{n+1}.
\end{align*}
By taking the explicit normalisation constant $h_n$ in \eqref{qorthogonal}, one can obtain
\begin{align*}
c_n=-(-\al)^{n+1}q^{n\choose 2} (q;q)_{n+1}.
\end{align*} 
The weight $\tilde{\omega}(x;q)$ of skew orthogonal Al-Salam $\&$ Carlitz polynomials is determined by
\begin{align*}
\tilde{\omega}(x;q)=f(qx;q)\omega(qx;q)=\al\omega(x;q),
\end{align*}
where $\omega(x;q)$ is the weight of Al-Salam $\&$ Carlitz polynomials defined by \eqref{wAC}.

\item Little $q$-Jacobi polynomials. For the weight function \eqref{weight}, the Pearson pair is
\begin{align*}
(f,g)=\left(-x^2+x,-q^{\frac{1}{2}}([\al+\be+2]_qx-[\al+1]_q)\right),
\end{align*}
which degenerates to the Pearson pair of Jacobi polynomials as $q\to 1^{-}$. 
With the help of the leading term in \eqref{4.24a}, one can obtain 
\begin{align*}
c_n=\left(q^n[\al+\be+2]_q+[n]_q-[-n]_q\right)h^{(\al,\be)}_{n+1}=q^{-n}[2n+\al+\be+2]_qh_{n+1}^{\ab},
\end{align*}
where $h_n^{(\al,\be)}$ is the normalisation constant of orthogonal relation \eqref{qj}. We remark that in this case, the weight of skew orthogonal little $q$-Jacobi polynomials is taken as $$\tilde{\omega}(x;q):=f(qx;q)\omega(qx;q)=-(qx)^{\al+1}(qx;q)_{\be+1}=-q^{\al+1}\omega^{(\al+1,\be+1)}(x;q)$$ with $\omega^\ab(x;q)$ being the weight of little $q$-Jacobi polynomials defined in \eqref{weight}. As we consider the weight $\omega^{(\al+1,\be+1)}(x;q)$ in Section \ref{lqj}, there is only a scalar transformation between it and $\tilde{\omega}(x;q)$.
\end{enumerate}

 \subsection{Kernels between $q$-unitary ensemble and $q$-symplectic ensemble}
 This part is devoted to the relationship between the kernels of $q$-unitary ensemble and $q$-symplectic ensemble. To this end, we firstly need to modify the inner product in \eqref{oip} and delta function given in \eqref{delta}.

Consider the symmetric inner product  (as above we assume the interval $I=[0,\infty)$; the other cases can be done similarly)
\begin{align*}
\langle \phi(x;q),\psi(x;q)\rangle=(1-q)\sum_{i\in\mathbb{Z}}\phi(q^i)\psi(q^i)\rho(q^i)q^i=\int_0^\infty \phi(x;q)\psi(x;q)\rho(x;q)d_qx.
\end{align*} 
Under this symmetric inner product, we can define a family of $q$-orthogonal polynomials $\{p_n(x;q)\}_{n=0}^\infty$ such that $\langle p_n(x;q),p_m(x;q)\rangle=\tih_n(q)\delta_{n,m}$ with some proper normalisation constant $\tih_n$. In this case, we can define the  delta function on $\mathcal{H}\to\mathcal{H}$ as
\begin{align*}
\delta(x,y;q)=\sum_{n=0}^\infty \frac{p_n(x;q)p_n(y;q)}{\tih_n(q)}, \end{align*}
such that $\langle \delta(x,y;q),\xi(x;q)\rangle=\xi(y;q),
$
and a finite dimensional projection operator
\begin{align*}
K_N(x,y;q)=\sum_{n=0}^{N-1} \frac{p_n(x;q)p_n(y;q)}{\tih_n(q)},
\end{align*}
which is the Christoffel-Darboux kernel of the $q$-unitary ensemble.

Consider now the equation $\langle \delta(x,y;q),\psi(x;q)\rangle_{s,\rho(qx)f(qx)}$.  According to the equation \eqref{qos} and the property of the delta function, one can get
\begin{align*}
\langle \delta(x,y;q),\psi(x;q)\rangle_{s,\rho(qx)f(qx)}=\langle \delta(x,y;q),\mathcal{A}_q\psi(x;q)\rangle=\mathcal{A}_q\psi(y;q).
\end{align*}
On the other hand, if we take $\psi(x;q)$ as the $q$-skew orthogonal polynomials $Q_{2m}(x;q)$, which is skew orthogonal with the weight $\rho(qx)f(qx)$, then 
\begin{align*}
\langle \delta(x,y;q), Q_{2m}(x;q)\rangle_{s,\rho(qx)f(qx)}=\sum_{n=0}^\infty \frac{p_n(y;q)}{\tih_n(q)}\langle p_n(x;q),Q_{2m}(x;q)\rangle_{s,\rho(qx)f(qx)}.
\end{align*}
By using the equality \eqref{qpq} and the skew orthogonality \eqref{qsop}, it follows from this that
\begin{align}\label{qeven}
\mathcal{A}_q Q_{2m}(y;q)=-\frac{u_m(q)}{\tih_{2m+1}(q)}p_{2n+1}(y;q)
\end{align}
(cf.~(\ref{dqeven})).
And if one takes $\psi(x;q)=Q_{2n+1}(x;q)$,  by again using the equations \eqref{qsop} and \eqref{qpq}, we obtain
\begin{align}\label{qodd}
\mathcal{A}_qQ_{2m+1}(y;q)=u_m(q)\left(\frac{p_{2m}(y;q)}{\tih_{2m}(q)}-\frac{c_{2m+1}(q)}{c_{2m}(q)}\frac{p_{2m+2}(y;q)}{\tih_{2m+2}(q)}\right)
\end{align}
(cf.~(\ref{dqodd})).

To make use of the above formulae,
first rewrite the Christoffel-Darboux kernel for the $q$-unitary ensemble in the case $N \mapsto 2N - 1$ as
\begin{align*}
K_{2N-1}(x,y;q)=\sum_{n=0}^{N-1}\frac{1}{\tih_{2n}(q)}p_{2n}(x;q)p_{2n}(y;q)+\sum_{n=0}^{N-1}\frac{1}{\tih_{2n+1}(q)}p_{2n+1}(x;q)p_{2n+1}(y;q).
\end{align*}
Making use of \eqref{qpq}, \eqref{qeven} and \eqref{qodd} shows
\begin{align*}
\sum_{n=0}^{N-1}\frac{1}{\tih_{2n+1}(q)}p_{2n+1}(x;q)p_{2n+1}(y;q)&=-\sum_{n=0}^{N-1}\frac{1}{u_n(q)}Q_{2m+1}(x;q)\mathcal{A}_qQ_{2m}(y;q),\\
\sum_{n=0}^{N-1}\frac{1}{\tih_{2n}(q)}p_{2n}(x;q)p_{2n}(y;q)&=\sum_{n=0}^{N-1}\frac{1}{u_n(q)}Q_{2m}(x;q)\mathcal{A}_q Q_{2m+1}(y;q)\\&+\frac{c_{2N-1}(q)}{\tih_{2N}(q)c_{2N-2}(q)}p_{2N}(y;q)Q_{2N-2}(x;q),
\end{align*} 
and hence 
\begin{align}\label{AR}
K_{2N-1}(x,y;q)=\mathcal{A}_{q,y}R_N(x,y;q)+\frac{c_{2N-1}(q)}{\tih_{2N}(q)c_{2N-2}(q)}p_{2N}(y;q)Q_{2N-2}(x;q).
\end{align}
Inspired by the continuous case \cite{adler00} and discrete case \cite{borodin09}, we expect to find an inverse operator of $\mathcal{A}_q$, allowing the kernel of $q$-symplectic ensemble be written as  the kernel of the $q$-unitary ensemble plus a rank $1$ decomposition, dependent on the $q$-orthogonal polynomials only.

For this purpose we note the operator $\mathcal{A}_q$, as defined in \eqref{qa},  can be also written as 
\begin{align*}
\mathcal{A}_q=\frac{\rho^{-1}(x;q)}{(q-1)x}\left[\rho(qx;q)f(qx;q)T_q-f(x;q)\rho(x;q)T_{q^{-1}}\right]:=\frac{\rho^{-1}(x;q)}{(q-1)x}\mathcal{R}_q
\end{align*}
with the shift operator $T_q\phi(x)=\phi(qx)$ and $T_{q^{-1}}\phi(x)=\phi(q^{-1}x)$. We see by inspection that the operator 
$\mathcal{R}_q$ has an inverse operator $\epsilon_q$ specified by
\begin{align}\label{qe}
\begin{aligned}
&(\epsilon_q\cdot\phi)(aq^{2m})=-\sum_{k=m}^{+\infty}\frac{\omega(aq^{2m+2};q)\cdots\omega(aq^{2k};q)}{\omega(aq^{2m+1};q)\cdots\omega(aq^{2k+1};q)}\phi(aq^{2k+1}),\\
&(\epsilon_q\cdot\phi)(aq^{2m+1})=\sum_{k=-\infty}^m\frac{\omega(aq^{2k+2};q)\cdots\omega(aq^{2m};q)}{\omega(aq^{2k+1};q)\cdots\omega(aq^{2m+1};q)}\phi(aq^{2k}),
\end{aligned}
\end{align}
where $\omega(x;q)=f(x;q)\rho(x;q)$ and $a\in\mathbb{C}^{\times}$. Hence we can write the inverse of operator $\mathcal{A}_q$ as $\mathcal{D}_q:=\epsilon_q (q-1)x\rho(x)$, 
allowing us to solve (\ref{AR}) for $R_N(x,y;q)$ (cf.~Proposition \ref{P3.12}).

\begin{proposition}\label{P4.7}
The Christoffel-Darboux kernel of the $q$-symplectic ensemble can be written in  terms of the  Christoffel-Darboux kernel of the $q$-unitary ensemble according to
\begin{align*}
R_N(x,y;q)=\mathcal{D}_{q,y}K_{2N-1}(x,y;q)-\frac{c_{2N-1}(q)u_{N-1}(q)}{c_{2N-2}(q)\tih_{2N}(q)\tih_{2N-1}(q)}(\mathcal{D}_q\cdot p_{2N})(y;q)(\mathcal{D}_q\cdot p_{2N-1})(x;q).
\end{align*}
\end{proposition}\label{P4.7}
The kernel of skew orthogonal Al-Salam $\&$ Carlitz polynomials can then be formulated as
\begin{align*}
R_N^{(AC)}(x,y)=\mathcal{D}_{q,y}K_{2N-1}^{(AC)}(x,y)+\frac{q^{1-2N}}{(1-q)\tih^{(AC)}_{2N-1}}(\mathcal{D}_q\cdot p_{2N})(y;q)(\mathcal{D}_q\cdot p_{2N-1})(x;q),
\end{align*}
and the one of skew orthogonal little $q$-Jacobi polynomials can be written as
\begin{align*}
&R_N(x,y)=\mathcal{D}_{q,y}K_{2N-1}(x,y)-\frac{q^{1-2N}[2N+\al+\be+1]_q}{\tih^{(\al,\be)}_{2N-1}}(\mathcal{D}_q\cdot p_{2N})(y;q)(\mathcal{D}_q\cdot p_{2N-1})(x;q).
\end{align*}

\section{Concluding remarks}
Consideration of discretisations of the eigenvalue PDF for Hermitian ensembles with symplectic symmetry lead to the
skew symmetric inner products (\ref{si}) (linear lattice) and (\ref{sinp}) (exponential lattice). We have presented here a theory of the corresponding
discrete, and $q$, skew orthogonal polynomials in the cases that the weight function in the corresponding discretised unitary ensemble
PDFs are classical. The theory is based on identifying the appropriate analogues of the differential operator (\ref{1.12c}), and developing their
properties, and direct analogues of the formulas (\ref{1.13a}) known in the continuous case are found.

In the classical, continuous cases it is also known \cite{adler00} that we have
\begin{equation}\label{6.1}
\Big ( f(x) w(x) \Big )^{1/2} Q_{2j}(x) = - {q_m \over h_{2m+1}} \int_x^\infty \Big ( {w(t) \over f(t) } \Big )^{1/2} p_{2j+1}(t) \, dt.
\end{equation}
An analogue of this formula does not appear possible in the discrete setting. It is (\ref{6.1}) which leads to the kernel summation
formula (\ref{1.13b}) known for the continuous case \cite{adler00}. In the discrete cases, there is no known analogue of
(\ref{1.13b}). Rather we have the less explicit kernel summations of Proposition \ref{P3.12} (first derived in  \cite{borodin09})
and Proposition \ref{P4.7}.

The question of deducing the analogue of the kernel summation in Proposition \ref{P3.12} on a linear lattice discretisation of
(\ref{1.4}) with $\beta = 1$ (orthogonal symmetry case) has been solved in \cite{borodin09}. The question of developing
a theory of the corresponding skew orthogonal polynomials remains to be addressed (one application of knowledge of the corresponding
normalisations would be to give an explanation of the (some of) the multiple summations found in
\cite{BKW17}), as does the task of developing a theory for the $q$ analogue of the orthogonal symmetry case.

\section*{Acknowledgement}
This work is part of a research program supported by the Australian Research Council (ARC) through the ARC Centre of Excellence for Mathematical and Statistical frontiers (ACEMS). PJF also acknowledges partial support from ARC grant DP170102028.

\small
\bibliographystyle{abbrv}

\begin{thebibliography}{1}

\bibitem{adler00}
M. Adler, P. J. Forrester, T. Nagao and P. van Moerbeke,
\emph{Classical skew orthogonal polynomials and random matrices},
J. Stat. Phys., \textbf{99} (2000) 141-170.

\bibitem{adler99}
M. Adler, E. Horozov and P. van Moerbeke,
\emph{The Pfaff lattice and skew-orthogonal polynomials},
Int. Mat. Res. Not., \textbf{11} (1999), 569-588.

\bibitem{aomoto98}
K. Aomoto,
\emph{On elliptic product formulas for Jackson integrals associated with reduced root systems},
J. Algebraic Combin., \textbf{8} (1998), 115-126.

\bibitem{askey80} 
R. Askey,
\emph{Some basic hypergeometric extensions of integrals of Selberg and Andrews},
SIAM J. Math. Anal., \textbf{11} (1980), 938-951.

\bibitem{AS93} 
R. Askey and S. K. Suslov. The $q$-harmonic oscillator and the Al-Salam and Carlitz polynomials.
Lett. Math. Phys., \textbf{29} (1993), 123--132.

\bibitem{baker00}
T. Baker and P. J. Forrester,
\emph{Multivariable Al-Salam $\&$ Carlitz polynomials associated with the type A $q$-Dunkl kernel},
Math. Nachr., \textbf{212} (2000), 5-35.

\bibitem{Bo11}
A.~Borodin, \emph{Determinantal point processes}, The {O}xford {H}andbook of {R}andom {M}atrix
  {T}heory (G.~Akemann, J.~Baik, and P.~di~Francesco, eds.), Oxford University
  Press, Oxford, 2011, pp.~231--249.
  
  \bibitem{BP14}
  A.~Borodin and L.~Petrov, \emph{Integrable probability: from representation theory to Macdonald processes},
  Probab. Surveys, \textbf{11} (2014), 1--58.

\bibitem{borodin09}
A. Borodin and E. Strahov,
\textit{Correlation kernels for discrete symplectic and orthogonal ensembles},
Comm. Math. Phys., {\bf{286}} (2009), 933-977.

\bibitem{BKW17}
R.P.~Brent, C.Kattenthaler and O.~Warnaar,
\textit{Discrete analogues of Macdonald--Mehta integrals},
J.~Comb.~Th. \textbf{144} (2016), 80--138.

\bibitem{debruijn55}
N.G. de Bruijn,
\emph{On some multiple integrals involving determinants},
J. Indian Math. Soc., \textbf{19} (1955), 133--151.

\bibitem{bunch82}
J. Bunch,
\emph{A Note on the stable decomposition of skew-symmetric matrices},
Math. Comp. \text{38} (1982), 475-479.

\bibitem{chang18}
X. Chang, Y. He, X. Hu and S. Li,
\textit{Partial-skew-orthogonal polynomials and related integrable lattices with Pfaffian tau-functions},
Comm. Math. Phys. {\bf{364}} (2018), 1069-1119. 

\bibitem{Ey01}
B. Eynard, \textit{Asymptotics of skew-orthogonal polynomials}, J. Phys. A \textbf{34} (2001), 7591-7605.

\bibitem{forrester10}
P. J. Forrester,
\emph{Log-gases and random matrices}, Princeton University Press, Princeton, NJ, 2010.

\bibitem{forrester18}
P. J. Forrester,
\emph{Meet Andr\'eief, Bordeaux 1886, and Andreev, Kharkov 1882--1883.}
Random matrices: Theory and applications, DOI: 10.1142/S2010326319300018, 2018.


\bibitem{FMS11}
 P.J.~Forrester, S.N.~Majumdar and G.~Schehr, Non-intersecting {B}rownian walkers and Yang-Mills theory on the sphere,
Nucl.~Phys.~B {\bf 844} (2011), 500--526.

\bibitem{FNR06}
P.J.~Forrester, T.~Nagao and E.M.~Rains, \emph{Correlation functions for random involutions},
Int. Math. Research Notices, \textbf{2006} (2006), 89796.

\bibitem{FR02a}
P.J.~Forrester and E.M.~Rains, \emph{Correlations for superpositions and decimations of Laguerre and
   Jacobi orthogonal matrix ensembles with a parameter},
Probab. Theory Related Fields, \textbf{130} (2004),    518--576.
 
 
\bibitem{FR02b}
P.J. Forrester and E.M. Rains, \emph{Interpretations of some parameter
  dependent generalizations of classical matrix ensembles}, Prob. Theory
  Related Fields \textbf{131} (2005), 1--61.

\bibitem{FW08}
P.J. Forrester and S.O. Warnaar, \emph{The importance of the {S}elberg
  integral}, Bull. Am. Math. Soc. \textbf{45} (2008), 489--534
  
\bibitem{Go09}  
S. Ghosh, \textit{Skew-orthogonal polynomials and random matrix theory}, CRM Monograph Series, vol. 28, American
Mathematical Society, Providence, RI, 2009.

\bibitem{ismail05}
M. Ismail,
\textit{Classical and quantum orthogonal polynomials in one variable},
Encyclopedia of Mathematics and its applications, Vol. 98, Cambridge University Press, 2005.

\bibitem{ito17}
M. Ito and P. J. Forrester,
\emph{A bilateral extension of the q-Selberg integral},
Trans. Amer. Math. Soc., \textbf{369} (2017), 2843-2878.

\bibitem{Jo01}
K. Johansson, \emph{Discrete orthogonal polynomial ensembles and the Plancherel measure},
Annals of Math., \textbf{153} (2001), 259--296.

\bibitem{Jo02}
K. Johansson,  \emph{Non-intersecting paths, random tilings and random matrices},
Probab.Theory Relat. Fields,   \textbf{123} (2002), 225-280.

\bibitem{kadell88}
K. Kadell,
\emph{A proof of Askey's conjectured q-analogue of Selberg's integral and a conjecture of Morris},
SIAM J. Math. Anal. \textbf{19} (1988), 969-986.


\bibitem{koekoek96}
R. Koekoek and R. Swarttouw,
\textit{The Askey-scheme of hypergeometric orthogonal polynomials and its $q$-analogue},
arXiv: math/9602214, 1996.


\bibitem{LW16}
K.~Liechty and D.~Wang, \textit{Nonintersecting Brownian motions of the unit circle},
Ann. Probab. \textbf{44} (2016), 1134-1211

\bibitem{masuda91}
T. Masuda, K. Mimachi, Y. Nakagami, M. Noumi and K. Ueno,
\emph{Representations of the quantum group $SU_q(2)$ and the little $q$-Jacobi polynomials},
J. Func. Anal., \textbf{99} (1991), 357-386.

\bibitem{NF02}
T.~Nagao and P.J.~Forrester, \emph{Vicious random walkers and a discretization of Gaussian random
matrix ensembles}, Nucl. Phys. B620 (2002) 551--565.

\bibitem{nikiforov86}
A. F. Nikiforov and S. K. Suslov,
\textit{Classical orthogonal polynomials of a discrete variable on nonuniform lattices},
Lett. Math. Phys., \textbf{11} (1986), 27-34.

\bibitem{rains00}
E. Rains,
\emph{Correlation functions for symmetrized increasing subsequences},
arXiv:0006097, 2000.

\bibitem{widom99}
H. Widom,
\textit{On the relation between Orthogonal, Symplectic and Unitary matrix ensembles},
J. Stat. Phys., \textbf{94} (1999), 347-363.



\end{thebibliography}

\end{document}